\documentclass{vldb}
\usepackage{graphicx}
\usepackage{balance}  % for  \balance command ON LAST PAGE  (only there!)

\usepackage{url}
\usepackage{comment}
\usepackage{mathtools}
\usepackage{amstext}    % defines the \text command
\usepackage{array}      % for advanced column specification: use >{\command} for
\usepackage{textcomp}	% for \textperthousand
\newcommand{\set}[1]{\{ #1 \}}
\numberwithin{equation}{section}
\newtheorem{theorem}{Theorem}
\newdef{definition}{Definition}
\usepackage{algorithm}
\usepackage[noend]{algpseudocode} % also loading algorithmicx

\makeatletter
\newcommand{\Setlineno}[1]{\setcounter{ALG@line}{\numexpr#1-1}}
\makeatother

\algdef{SE}[IF]{NoThenIf}{EndIf}[1]{\algorithmicif\ #1}{\algorithmicend\
\algorithmicif}%
\newcommand{\hStatex}[0]{\vspace{5pt}}
\newcommand{\algvariable}[1]{\texttt{#1}}
\newcommand{\algkeyword}[1]{\textbf{#1}}
\begin{document}

\title{Almost Strong Consistency: ``Good Enough'' in Distributed Storage
Systems}

\numberofauthors{4}

\author{
\alignauthor
Hengfeng Wei\\
       \affaddr{State Key Laboratory for Novel Software Technology}\\
       \affaddr{Nanjing University, China}\\
       \email{hengxin0912@gmail.com}
\alignauthor
Yu Huang\\
       \affaddr{State Key Laboratory for Novel Software Technology}\\
       \affaddr{Nanjing University, China}\\
       \email{yuhuang@nju.edu.cn}
\alignauthor
Jiannong Cao\\
       \affaddr{Hong Kong Polytechnic University}\\
       \affaddr{Hong Kong, China}\\
       \email{csjcao@comp.polyu.edu.hk}
\and  % use '\and' if you need 'another row' of author names
\alignauthor
Jian Lu\\
       \affaddr{State Key Laboratory for Novel Software Technology}\\
       \affaddr{Nanjing University, China}\\
       \email{lj@nju.edu.cn}
}

\maketitle

\begin{abstract}
  A consistency/latency tradeoff arises as soon as a distributed storage system
  replicates data.
  For low latency, modern storage systems often settle for weak consistency
  conditions,
  which provide little, or even worse, no guarantee for data consistency.
  In this paper we propose the notion of \emph{almost strong consistency} as a
  better balance option for the consistency/latency tradeoff.
  It provides both deterministically bounded staleness of data versions
  for each \textsl{read} and probabilistic quantification on the rate of
  ``reading stale values'', while achieving low latency.
  In the context of distributed storage systems, we investigate almost strong
  consistency in terms of \mbox{\emph{2-atomicity}}.
  Our 2AM (\mbox{2-Atomicity} Maintenance) algorithm completes both
  \textsl{reads} and \textsl{writes} in \emph{one} communication round-trip,
  and guarantees that each \textsl{read} obtains the value of within the latest
  2 versions.
  To quantify the rate of ``reading stale values'', we decompose the
  so-called ``old-new inversion'' phenomenon into concurrency patterns and
  \mbox{read-write} patterns, and propose a stochastic queueing model and a
  timed balls-into-bins model to analyze them, respectively.
  The theoretical analysis not only demonstrates that ``old-new inversions''
  rarely occur as expected, but also reveals that the \mbox{read-write} pattern
  dominates in guaranteeing such rare data inconsistencies.
  These are further confirmed by the experimental results, showing that
  \mbox{2-atomicity} is ``good enough'' in distributed storage systems
  by achieving low latency, bounded staleness, and rare data inconsistencies.
\end{abstract}

\vspace{-0.7pt}
% \begin{IEEEkeywords}
%   Almost strong consistency, 2-atomicity, bounded staleness, quantification.
% \end{IEEEkeywords}
%%%%%%%%%%%%%%%%%%%%%%%%%%%%%%%%%%%%%%%%%%%%%%%%%
\section{Introduction} \label{section:introduction}

% data consistency problem
Distributed storage systems \cite{bigtable-Chang06-osdi}
\cite{amazon-Vogels07-sosp} \cite{pnuts-Cooper08-vldb}
\cite{facebook-Beaver10-osdi} are considered as integral and fundamental
components of modern Internet services such as e-commerce and social networks.
They are expected to be fast, always available, highly scalable, and
network-partition tolerant.
To this end, modern distributed storage systems typically replicate their data
across different machines and even across multiple data centers, at the expense
of introducing data inconsistency.

% consistency/latency tradeoff
More importantly, as soon as a storage system replicates data, a tradeoff
between consistency and latency arises \cite{Abadi12}.
This consistency/latency tradeoff arguably has been highly influential in
system design because it exists even when there are no network partitions
\cite{Abadi12}.
In distributed storage systems, latency is widely regarded as a critical factor
for a large class of applications.
For example, the experiments from Google \cite{Google09} demonstrate that
increasing web search latency 100 to 400 ms reduces the daily number of
searches per user by 0.2\% to 0.6\%.
Therefore, most storage systems (and applications built on them) are designed
for low latency in the first place.
They often sacrifice strong consistency and settle for weaker ones,
such as eventual consistency \cite{amazon-Vogels07-sosp}, per-record timeline
consistency \cite{pnuts-Cooper08-vldb}, and causal consistency \cite{Lloyd11}.
However, such weak consistency models usually provide little, or even worse, no
guarantee for data consistency.
More specifically, they neither make any deterministic guarantee on the
staleness of the data returned by \textsl{reads} nor provide probabilistic
hints on the rate of violations with respect to the desired strong consistency.

% propose the notion of almost strong consistency
In this paper we propose the notion of \emph{almost strong consistency} as a
better balance option for the consistency/latency tradeoff.
The implication of the term ``almost'' is twofold.
On one hand, it provides deterministically bounded staleness of
data versions for each \textsl{read}.
Thus, the users are confident that out-of-date data is still useful
as long as they can tolerate certain staleness.
On the other hand, it further provides probabilistic quantification on the
low rate of ``reading stale values''.
This ensures that the users are actually accessing up-to-date data most of the
time.

% example for ``almost strong consistency''
We illustrate the idea of almost strong consistency by an exemplar
mobile-app-based taxi transportation system.
In this system, each taxi periodically reports its location data to the data
server.
Due to the natural locality of the update and request of location data, the city
is partitioned into multiple areas and a data server is deployed in each area.
The location data is replicated among all the data servers.
Thus the users all over the city can request the location data via a mobile
application like Uber \cite{Uber}.
Though consistency is a desirable property, in this application the user may be
more concerned of how long he has to wait before his query can be served.
We argue that the application may be willing to trade certain consistency for
low latency, as long as the inconsistency is bounded and the application can
still access up-to-date data most of the time \cite{YuHuang-tpds10}.

% ``almost strong consistency'' in terms of 2-atomicity]}}

In the context of distributed storage systems, we investigate almost strong
consistency in terms of \mbox{\emph{2-atomicity}}.
By instantiating the abstract notion of almost strong consistency, the
2-atomicity semantics also includes two essential parts as elaborated below.

First, the \mbox{2-atomicity} semantics guarantees that the value returned by
each \textsl{read} is one of the latest 2 \emph{versions},
besides admitting an implementation with low latency.
By ``low latency'' we mean that both \textsl{reads} and \textsl{writes} complete
in \emph{one} communication round-trip.
Theoretically, it has been proved impossible to achieve low latency while
enforcing each \textsl{read} to return the \emph{latest} data version,
as required in atomicity \cite{interprocess-Lamport86-dc}
\cite{linearizability-Herlihy90-toplas}, given that a minority of replicas may
fail \cite{Dutta04}.
For example, the ABD algorithm \cite{abd-Attiya95-jacm} for emulating atomic
registers requires each \textsl{read} to complete in two \mbox{round-trips}.
This impossibility result justifies the relaxed consistency semantics of
2-atomicity.
In the transportation system example above, the taxi location data can still be
useful if the data returned is no more stale than the previous version to the
latest one.
This is mainly because the location data cannot change abruptly and the taxi
frequently updates its location in this scenario.

Second, the \mbox{2-atomicity} semantics provides probabilistic quantification
on the rate of violations of atomicity.
In data storage systems, atomicity is widely used as the formal definition of
strongly consistent or up-to-date data access.
By bounding the probability of violating atomicity, the \mbox{2-atomicity}
semantics provides another orthogonal perspective for expressing how strong
consistency is ``almost" guaranteed.
In our example above, since the user may request the location data of a number
of taxies, the inconsistency data may not affect the quality of service
experienced by the user, as long as only a small portion of the query return
slightly stale data.

Our 2AM (2-Atomicity Maintenance) algorithm for maintaining 2-atomicity in
distributed storage systems completes both \textsl{reads} and \textsl{writes} in
\emph{one} communication round-trip, and guarantees that each \textsl{read}
obtains the value of within the latest 2 versions.
To quantify the rate of ``reading stale values'', we decompose the so-called
``old-new inversion'' phenomenon into two patterns: concurrency pattern and
\mbox{read-write} pattern.
We then propose a stochastic queueing model and a timed balls-into-bins
model to analyze the two patterns, respectively.
The theoretical analysis not only demonstrates that ``old-new inversions''
rarely occur as expected, but also reveals that the \mbox{read-write} pattern
dominates in guaranteeing such rare violations.

We have also implemented a prototype data storage system among mobile phones,
which provides \mbox{2-atomic} data access based on the 2AM algorithm and atomic
data access based on the ABD algorithm.
The \textsl{read} latency in our 2AM algorithm has been significantly reduced,
compared to that in the ABD algorithm.
More importantly, the experimental results have confirmed our theoretical
analysis above.
Specifically, the proportion of old-new inversions incurred in the 2AM
algorithm is typically less than 0.1\text{\textperthousand}, and the proportion
of read-write patterns among concurrency patterns (e.g., about
0.1\text{\textperthousand} in some setting) is much less than that of
concurrency patterns themselves (e.g., more than $50\%$ in the same setting).
Thus \mbox{2-atomicity} is ``good enough'' in distributed storage systems
by achieving low latency, bounded staleness, and rare atomicity violations.

The remainder of the paper is organized as follows.
Section \ref{section:almost-strong-consistency} proposes the notion of almost
strong consistency and discusses how to define it in terms of
\mbox{2-atomicity} in the context of distributed storage systems.
Section \ref{section:achieve-asc} presents the 2AM (2-Atomicity
Maintenance) algorithm which achieves deterministically
bounded staleness.
Section \ref{section:quantifying} is concerned with the theoretical analysis of
the atomicity violations incurred in the 2AM algorithm.
Section \ref{section:experiment} presents the prototype data storage
system and experimental results.
Section \ref{section:related-work} reviews the related work.
Section \ref{section:conclusion} concludes the paper.
%%%%%%%%%%%%%%%%%%%%%%%%%%%%%%%%%%%%%%%%%%%%%%%%%
\section{Almost Strong Consistency} \label{section:almost-strong-consistency}

In this section, we propose the notion of \emph{almost strong consistency}, and
instantiate it in terms of \mbox{\emph{2-atomicity}}, in the context
of distributed storage systems.
%%%%%%%%%%%%%%%%%%%%%%%%
\subsection{Generic Notion of Almost Strong Consistency}
\label{subsection:generic-asc}

The distributed storage system consists of an arbitrary number of $N$
\emph{clients} and a fixed number $n$ of server replicas (or \emph{replicas},
for short) that communicate through message-passing (Figure \ref{fig:system-model}).
Each replica maintains a set of replicated key-value pairs
(also referred to as \emph{registers} in the sequel).

%%%%%%%%%%%%%%
\begin{figure}[!t]
  \centering
  \includegraphics[width = 0.30\textwidth]{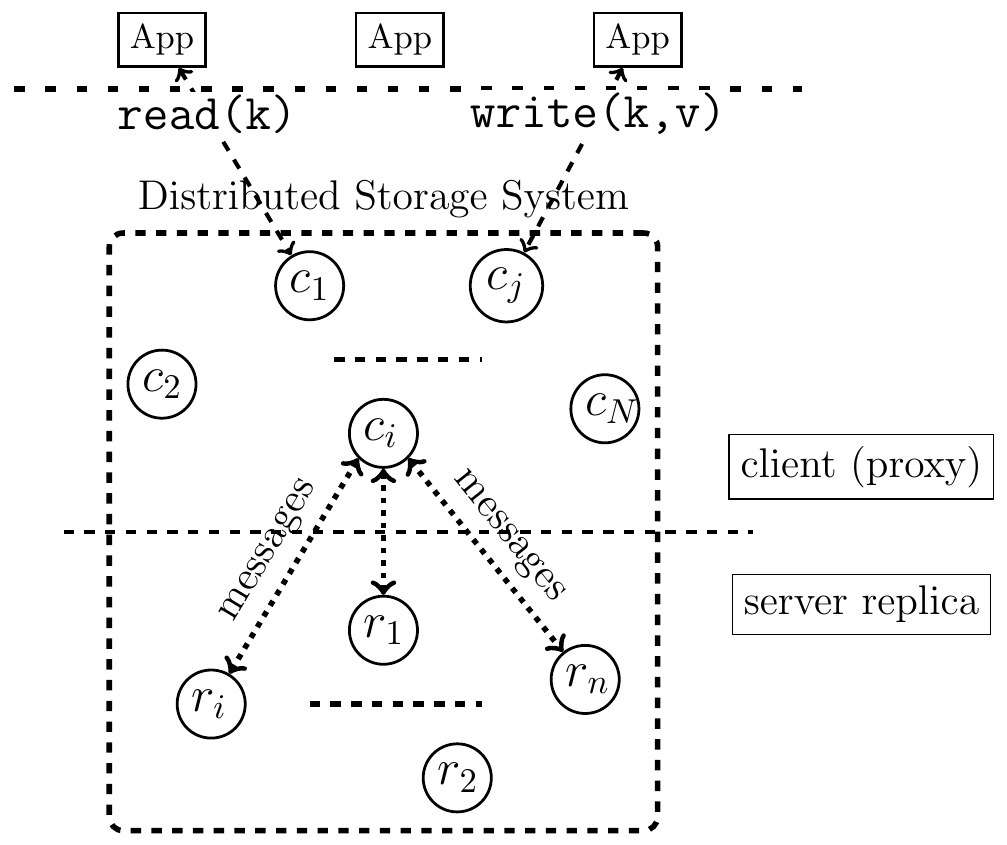}
  \caption[labelInTOC]{Distributed storage system model.}
  \label{fig:system-model}
\end{figure}
%%%%%%%%%%%%%%

The distributed storage system supports two \emph{operations} to upper-layer
applications:
\itshape 1)\upshape \; storing a value associated with a key, denoted
\textsl{write(key,value)}; \emph{and}
\itshape 2)\upshape \; retrieving a value associated with a key, denoted
\textsl{value $\gets$ read(key)}.
Clients serve as the proxies for applications by invoking \textsl{read/write}
operations on the registers and communicating with replicas on
behalf of them.
Being replicated, different versions of the same register may co-exist.
The concept of consistency models is then introduced to constrain the possible
data versions that are allowed to be returned by each \textsl{read}.
Particularly, strong consistency requires each \textsl{read} to obtain
the latest data version according to some sequential order.

The notion of \emph{almost strong consistency} generalizes the traditional
strong consistency by allowing stale data versions to be read. That is,
\begin{enumerate}
  \item It provides deterministically bounded staleness of data versions for
  each \textsl{read};
%   Thus, developers and end users are sure that out-of-date data is still useful
%   as long as the application can tolerate certain staleness.
  \item It provides probabilistic quantification on the rate of
  ``reading stale values''.
%   This makes it possible for people to claim, with high confidence, that their
%   data tend to be fresh.
\end{enumerate}
%%%%%%%%%%%%%%%%%%%%%%%%
\subsection{Almost Strong Consistency in Terms of 2-Atomicity}
\label{subsection:asc-2atomicity}

In the context of distributed storage systems, we investigate almost strong
consistency in terms of \mbox{\emph{2-atomicity}}.
As preliminaries, we first review atomicity \cite{dc-Attiya04-book}.
From the view of clients, each operation is associated with two events: an
\emph{invocation} event and a \emph{response} event.
For a \textsl{read} (on a specific key), the invocation is denoted
\textsl{read(key)}, and its response has the form \textsl{ack(value)}, returning
some value to the client.
For a \textsl{write}, the invocation is denoted \textsl{write(key, value)}, and
its response is an \textsl{ack}, indicating its completion.
We assume an imaginary global clock and all the events are time-stamped with
respect to it \cite{interprocess-Lamport86-dc}.
Among all the \textsl{writes}, we posit, for each register, the existence of a
special one which writes the \emph{initial value}, at the very beginning of the
imaginary global clock.
% We consider an asynchronous system, meaning that there is no real global time
% available to the processes themselves. However, for specification and
% correctness analysis, we employ an imaginary global time
% \cite{interprocess-Lamport86-dc} and assign time to events of invocations and
% responses.
%
% In a distributed storage system, all the server replicas (maybe together with
% the client processes) collectively implement the read and write operations,
% providing a shared memory illusion to clients on top of the message-passing
% system \cite{Li89}.
% To be useful, the implementations should conform to some consistency model.
% Roughly speaking, a consistency model defines a \emph{test} on the executions of
% the system.
% A system is said to satisfy some consistency model if all its executions do.
% In the rest of this subsection, we describe what an execution is and its
% properties.
% A review of atomicity \cite{interprocess-Lamport86-dc} based on executions is given in the
% following subsection.

An \emph{execution} $\sigma$ of the distributed storage system is a sequence of
invocations and responses.
An operation $o_1$ \emph{precedes} another operation $o_2$, denoted
$o_1 \prec_{\sigma} o_2$ (or $o_1 \prec o_2$ if $\sigma$ is clear or
irrelevant), if and only if the response of $o_1$ occurs in $\sigma$ before the
invocation of $o_2$.
Two operations are considered \emph{concurrent} if neither of them precedes the other.
An execution $\sigma$ is said \emph{well-formed} if each client invokes at most
one operation at a time, that is, for each client $p_i$, $\sigma|i$ (the
subsequence of $\sigma$ restricted on $p_i$) consists of alternating invocations
and matching responses, beginning with an invocation.
A well-formed execution $\sigma$ is \emph{sequential} if for each operation
in $\sigma$, its invocation is immediately followed by its response.

% Furthermore, a sequential execution $\sigma$ is \emph{legal} if each
% \textsl{read} on a specific key returns the value of the most recently
% preceding \textsl{write} in $\sigma$ on the same key, if there is one, and
% otherwise returns its initial value.

Intuitively, atomicity requires each operation to appear to take effect
instantaneously at some point between its invocation and its response.
More precisely,

\begin{definition}	\label{def:atomicity}
  A storage system satisfies \textbf{atomicity} \cite{dc-Attiya04-book} if, for
  each of its well-formed executions $\sigma$, there exists a permutation $\pi$
  of all the operations in $\sigma$ such that $\pi$ is sequential and
  \begin{itemize}
    \item $[$real-time requirement$]$ If $o_1 \prec_{\sigma} o_2$, then $o_1$
    appears before $o_2$ in $\pi$; \emph{and}
    \item $[$read-from requirement$]$ Each \textsl{read} returns the value
    written by the most recently preceding \textsl{write} in $\pi$ on the same
    key, if there is one, and otherwise returns its initial value.
  \end{itemize}
\end{definition}

The semantics of 2-atomicity is adapted from that of atomicity by relaxing its
read-from requirement to allow stale values to be read.

\begin{definition}	\label{definition:2-atomicity}
  A storage system satisfies \textbf{2-atomicity} if, for each of its
  well-formed executions $\sigma$, there exists a permutation $\pi$ of all the
  operations in $\sigma$ such that $\pi$ is sequential and
  \begin{itemize}
    \item $[$real-time requirement$]$ If $o_1 \prec_{\sigma} o_2$, then $o_1$
    appears before $o_2$ in $\pi$; \emph{and}
    \item $[$weak read-from requirement$]$ Each \textsl{read} returns the value
    written by one of the latest two preceding \textsl{writes} in $\pi$ on the
    same key.
  \end{itemize}
\end{definition}

In terms of \mbox{2-atomicity}, the notion of almost strong consistency can
then be instantiated as follows.
\begin{enumerate}
  \item Besides admitting an implementation with low latency, it
  guarantees that each \textsl{read} obtains the value of within the latest
  2 versions;
  \item It provides probabilistic quantification on the rate of
  actually reading the stale data version.
\end{enumerate}

Sections~\ref{section:achieve-asc} and \ref{section:quantifying} are concerned
with these two aspects, respectively.
%%%%%%%%%%%%%%%%%%%%%%%%%%%%%%%%%%%%%%%%%%%%%%%%%
\section{Achieving 2-Atomicity}
\label{section:achieve-asc}

In this section, we present the 2AM (\mbox{2-Atomicity} Maintenance) algorithm
for emulating \mbox{2-atomic}, Single-Writer Multi-Reader (SWMR)
registers.
It completes both \textsl{reads} and \textsl{writes} in \emph{one} round-trip,
and guarantees that each \textsl{read} obtains the value of within the latest 2
versions.

Despite its simplicity, SWMR registers are useful in a wide range of
applications, especially where the shared data has its natural ``owner".
Moreover, multiple SWMR registers can be used in group.
The typical setting is that each process has its ``own" register i.e., only the
owner process can write this register, while all processes can read all registers.
Multiple processes can communicate with each other by writing its own register
and reading other registers.
A possible alternative is to use Multi-Writer Multi-Reader (MWMR, for short)
registers.
Compared with using MWMR registers, using SWMR register in group may be more
compatible with the application logic, and the implementation is less complex
and has better maintainability.

%%%%%%%%%%%%%%%%%%%
\subsection{The 2AM (2-Atomicity Maintenance) Algorithm}
\label{subsection:algorithm-2atomicity}

We use the asynchronous, non-Byzantine model, in which:
\itshape 1)\upshape \;Messages can be delayed, lost, or delivered out of order,
but they are not corrupted; \emph{and}
\itshape 2)\upshape \;An arbitrary number of clients may crash while only a
minority of replicas may crash.

The 2AM algorithm is an adaptation from that for atomicity
\cite{abd-Attiya95-jacm}.
It makes use of versioning.
Specifically, for each \textsl{write(key, value)}, the writer associates a
\emph{version} with the key-value pair.
Each replica replaces a key-value pair it currently holds whenever a larger
version with the same key is received.
When reading from a key, a client tries to retrieve the value with the largest
version.
Since there is only one writer, versions (for each key) can be chosen totally
ordered using its local sequence numbers.

At its core, the algorithm is stated in terms of the \emph{majority quorum
systems} in the way that each operation is required to contact any majority of
the replicas to proceed. Specifically,

\begin{itemize}
  \item \textsl{write(key, value):} To write a value on a specific
  key, the single writer first generates a larger version than those it has ever
  used, associates it with the key-value pair, sends the versioned
  key-value pair to all the replicas, and waits for acknowledgments from a
  majority of them.
  \item \textsl{read(key):} To read from a specific key, the reader
  first queries and collects a set of versioned key-value pairs from a majority
  of the replicas, from which it chooses the one with the largest version to
  return.
\end{itemize}

As mentioned before, each replica replaces its key-value pair whenever a larger
version with the same key from a \textsl{write} is received. Besides, it
responds to the queries from \textsl{reads} with the versioned key-value pair
it currently holds.

The pseudo-code for \textsl{read} and \textsl{write} operations and the
replicas appears in Algorithm~\ref{alg:2-atomicity}.
Notice that the \textsl{read} here does not spend a second \mbox{round-trip}
propagating the returned value (along with its version) to a majority of the
replicas, in contrast to that in \cite{abd-Attiya95-jacm}.
The second round-trip in \cite{abd-Attiya95-jacm} (often referred to as the
``write back'' phase) is required to avoid the ``old-new inversion''
phenomenon \cite{Dutta04}, \cite{Attiya10}.
An ``old-new inversion'' witnesses a violation of atomicity, where two
non-overlapping \textsl{reads}, both overlapping a \textsl{write}, obtain
out-of-order values.
In the 2AM algorithm, we have intentionally ignored the ``write
back'' phase.
In the following subsection, we prove that the 2AM algorithm
indeed achieves the emulation of \mbox{2-atomic}, single-writer multi-reader
registers.
%%%%% Algorithm for 2-atomicity %%%%%
\begin{algorithm}[!t]
  \caption{The 2AM (2-Atomicity Maintenance) algorithm for \mbox{2-atomic},
  single-writer multi-reader registers.}
  \label{alg:2-atomicity}
  \begin{algorithmic}[1]
    \Procedure{write}{\algvariable{key,value}} \Comment{for the writer}
      \State increment \algvariable{version} for this \algvariable{key}

      \hStatex
      \State \algkeyword{pfor each} replica $s$	\Comment{\algkeyword{pfor} is a
      parallel for}
      \State \quad send $[\textsc{update} \algvariable{,key,value,version}]$ to
      $s$
      \State \algkeyword{wait for} $[\textsc{ack}]$s from a majority of replicas
    \EndProcedure

    \hStatex
    \Setlineno{1}
    \Procedure{read}{\algvariable{key}}	\Comment{for each reader}
      \State $results \leftarrow \emptyset$

      \hStatex
      \State \algkeyword{pfor each} replica $s$
        \State \quad send $[\textsc{query} \algvariable{,key}]$ to $s$
        \State \quad obtain $result \leftarrow
        [\algvariable{k,val,ver}]$ from $s$
        \State \quad $results \leftarrow results \cup \{
        result \}$ \State \algkeyword{until} a majority of replicas respond
      \hStatex

      \State \Return \algvariable{val} with the largest \algvariable{ver}
      in $results$
    \EndProcedure

    \hStatex
    \Statex $[\algvariable{k,val,ver}]:$ local versioned key-value pairs

    \hStatex
	\noindent $\triangleright$ The following procedure is executed in an
	uninterrupted way.
	Assume that \algvariable{msg} is from client $p_i$.
    \Setlineno{1}
    \Procedure{Upon}{\algvariable{msg}} \Comment{for each replica}
      \NoThenIf{\algvariable{msg} \algkeyword{instanceof} $[\textsc{query}
      \algvariable{,key}]$}
      \State \hspace{-8pt} send $[\algvariable{k,val,ver}]$ with
      $\algvariable{k} = \algvariable{key}$ to client $p_i$
      \EndIf
      \NoThenIf{\algvariable{msg} \algkeyword{instanceof} $[\textsc{update}
      \algvariable{,key,value,version}]$}
        \State pick $[\algvariable{k,val,ver}]$ with $\algvariable{k}
        = \algvariable{key}$
        \NoThenIf{$\algvariable{ver} < \algvariable{version}$}
          \State $\algvariable{val} \leftarrow \algvariable{value}$
          \State $\algvariable{ver} \leftarrow \algvariable{version}$
        \EndIf
        \State send [\textsc{ack}] to client $p_i$
      \EndIf
    \EndProcedure
  \end{algorithmic}
\end{algorithm}
%%%%% End: algorithm for 2-atomicity %%%%%

%%%%%%%%%%%%%%%%%%%%%%%%
\subsection{Correctness Proof of the 2AM Algorithm}
\label{subsection:correctness}

We aim to prove that, in the 2AM algorithm, the value returned
by each \textsl{read} is of one of the latest 2 versions.
It is basically a case-by-case analysis, concerning the partial order among and
the semantics of the \textsl{read/write} operations.
% More importantly, this correctness proof has identified the necessary and
% sufficient condition for the ``old-new inversion" phenomenon.
% This condition lays the foundation for the probabilistic analysis of the
% violations of atomicity in the next section (Section \ref{section:quantifying}).

\begin{theorem}	\label{theorem:2-atomicity}
  The 2AM algorithm achieves the emulation of \mbox{2-atomic}, single-writer
  multi-reader registers.
\end{theorem}

\begin{proof}
  First of all, we notice that \mbox{2-atomicity}, like atomicity, is a local property
  \cite{linearizability-Herlihy90-toplas}. Therefore, we can prove the
  correctness of the 2AM algorithm by reasoning independently about each
  individual register accessed in an execution.
  Without loss of generality, we assume that all the operations involved in
  the following correctness proof are performed on the same register.

  According to the definition of \mbox{2-atomicity}
  (Definition~\ref{definition:2-atomicity}), it suffices to identify a
  permutation $\pi$ of any execution of the 2AM algorithm, and to prove that
  $\pi$ is sequential and satisfies both the ``real-time requirement'' and the
  ``weak read-from requirement''.

  For any execution $\sigma$, we obtain its permutation $\pi$ in the following
  manner:
  \begin{itemize}
    \item All the \textsl{write} operations issued by the single writer are
    totally ordered according to the versions they use.
    \item The \textsl{read} operations are scheduled one by one in order of
    their \emph{invocation time:}
    A \textsl{read} $r$ that reads from a \textsl{write} $w$ is scheduled
    immediately after both $w$ and all the \textsl{read} operations preceding
    $r$ in the sense of $\prec_{\sigma}$ (which have already been scheduled).
  \end{itemize}

  Obviously, this permutation $\pi$ is sequential and satisfies the ``real-time
  requirement'' of \mbox{2-atomicity}.
  It remains to show that it satisfies the ``weak read-from requirement'' for
  each \textsl{read} as well.
  This argument involves a case-by-case analysis, concerning the partial order
  among and the semantics of the \textsl{read/write} operations.

  Here and in the sequel, we use the following notations:
  For an operation $o$, let $o_{st}$ denote its \emph{start time} (i.e., the
  time of its invocation event), $o_{ft}$ its \emph{finish time} (i.e., the time
  of its response event), and $[o_{st}, o_{ft}]$ its \emph{time interval}
  (Figure \ref{fig:old-new-inversion} for an example).
  We also write $r = R(w)$ to denote the \emph{``read-from''} relation in which
  the \textsl{read} $r$ reads from the \textsl{write} $w$.

  For any \textsl{read} operation $r$, we consider two cases according to
  whether there are concurrent \textsl{write} operations with it in the
  execution $\sigma$.

  \textsc{Case 1:} \emph{There is no concurrent \textsl{write} with $r$}.
  According to the 2AM algorithm (Algorithm~\ref{alg:2-atomicity}), especially
  due to the mechanism of the majority quorum systems, the \textsl{read} $r$
  must read from its most recently preceding \textsl{write} $w$, and hence in
  $\pi$, it is scheduled between $w$ and the next \textsl{write}.

  \textsc{Case 2:} \emph{There are concurrent \textsl{writes} with $r$}, among
  which the leftmost one is denoted $w$. Notice that $r_{st} \in [w_{st},
  w_{ft}]$ holds for $w$. There are two sub-cases according to the
  \textsl{write} from which $r$ reads.

  \textsc{Case 2.1:} \emph{$r$ reads from some concurrent \textsl{write}}.
  In this case, $r$ is scheduled in $\pi$ between this \textsl{write} and its
  next one.

  \textsc{Case 2.2:} \emph{$r$ reads from its most recently preceding
  \textsl{write} in $\sigma$ (denoted $w'$)}. Notice that \textsc{Case 2.1} and
  \textsc{Case 2.2} are exhaustive since $r$ cannot read from any earlier
  \textsl{writes} than $w'$ due to the mechanism of the majority quorum systems.
  To form an ``old-new inversion'', there must be at least two \textsl{read}
  operations.
  Therefore, in \textsc{Case 2.2} (shown in Figure~\ref{fig:old-new-inversion})
  we now consider other \textsl{read} operations (than $r$).

  \textsc{Case 2.2.1:} \emph{There is no \textsl{read} $r'$ that precedes $r$ in
  $\sigma$ and is concurrent with $w$}. Formally, $\nexists r' : r'_{ft} \in
  [w_{st}, r_{st}]$.
  In $\pi$, $r$ is scheduled between $w'$ and its next \textsl{write} (i.e.,
  $w$).

  \textsc{Case 2.2.2:} \emph{There is some \textsl{read} $r'$ that precedes $r$
  and is concurrent with $w$}. Formally, $\exists r': r'_{ft} \in [w_{st},
  r_{st}]$.

  Furthermore, if $r'$ reads from $w$, we obtain an ``old-new inversion'',
  where two non-overlapping \textsl{reads} (i.e., $r$ and $r'$), both
  overlapping a \textsl{write} (i.e., $w$), obtain out-of-order values.
  The scenario is depicted in Figure~\ref{fig:old-new-inversion}, where the
  dotted, directed arrows denote the read-from relation.

  In this situation, both $r$ and $r'$ are scheduled in $\pi$ between $w$ and
  its next \textsl{write}.
  As a consequence, $r$ reads from $w'$ which is its \emph{second} most recently
  preceding \textsl{write} in $\pi$, meeting the ``weak read-from requirement''
  of \mbox{2-atomicity}.

  Notice that \textsc{Case 2.2.2} (and thus the ``old-new inversion''
  phenomenon) is the \emph{only} case which leads to the violations of
  atomicity.
\end{proof}

%%%%% Figs for ``old-new inversion'' (aligned vertically) %%%%%
\begin{figure}[!t]
  \centering
    \includegraphics[width = 0.35\textwidth]
    {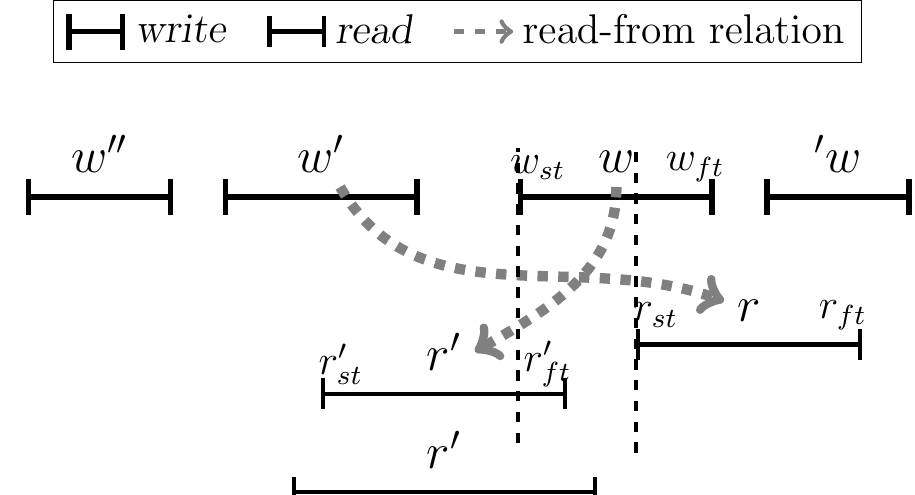}
    \caption[labelInTOC]{Old-new inversion. Two non-overlapping \textsl{reads}
    $r$ and $r'$, both overlapping the \textsl{write} $w$, obtain out-of-order
    values.
    (\emph{Time goes from left to right}.)}
    \label{fig:old-new-inversion}
\end{figure}
%%%%% End: figure for ``old-new inversion'' %%%%%

% An ``old-new inversion'' of \textsc{Case} 2.2.2 (shown in Figure
% \ref{subfig:oni}) does happen \cite{Attiya10} where $\ldots$

% We first comment that 2-atomicity (with single-writer) has also been studied
% in \cite{Aiyer05}, as a special case of $k$-atomicity.
% However, its algorithm requires each \textsl{read} to complete in two
% round-trips.
% Second, for the single-writer case, 2-atomicity turns out to be equivalent to
% regularity \cite{interprocess-Lamport86-dc}.
% Finally, distributed storage systems \cite{Vogels09},
% \cite{cassandra-Lakshman10} that employ Algorithm~\ref{alg:2-atomicity} (or more
% generally, the communication pattern $R + W > N$ given $N$ replicas and
% \textsl{read} and \textsl{write} quorum sizes $R$ and $W$) are .

%%%%%%%%%%%%%%%%%%%%%%%%%%%%%%%%%%%%%%%%%%%%%%%%%
\section{Quantifying the Atomicity Violations}
\label{section:quantifying}

In this section, we quantify the atomicity violations incurred in the 2AM
algorithm.
It follows from the correctness proof in Section \ref{subsection:correctness} that
the atomicity violations are exactly characterized by the ``old-new inversions"
in \textsc{Case 2.2.2}.
Furthermore, the proof has also identified the necessary and sufficient
condition for the ``old-new inversions'' phenomenon.
We formally define it as follows.

\begin{definition} \label{def:oni}
  The \textbf{old-new inversion} involving a \textsl{read} $r$ consists of the
  \textsl{read} $r$, two \textsl{writes} $w$ and $w'$, and a second
  \textsl{read} $r'$, such that (see Figure~\ref{fig:old-new-inversion})

  \itshape 1) \upshape $r_{st} \in [w_{st}, w_{ft}]$,
  \itshape 2) \upshape $w'$ immediately precedes $w$: $w' \prec w$, and no other
    \textsl{writes} are between $w$ and $w'$,
  \itshape 3) \upshape $r'_{ft} \in [w_{st}, r_{st}]$,

  \itshape 4) \upshape $r = R(w')$, \emph{and}
  \itshape 5) \upshape $r' = R(w)$.
\end{definition}

The five requirements for ``old-new inversion''  fall into two categories.
The first three requirements involve the partial order $\prec$ on,
and thus the \emph{concurrency patterns} among, \textsl{read/write} operations.
Intuitively, the higher degree of concurrency an execution shows, the more
``old-new inversions'' it may produce.

\begin{definition} \label{def:concurrency-pattern}
The \textbf{concurrency pattern} involving a \\ \textsl{read} $r$ consists of
the \textsl{read} $r$, two \textsl{writes} $w$ and $w'$, and a second
\textsl{read} $r'$, such that
  \begin{enumerate}
    \item $r_{st} \in [w_{st}, w_{ft}]$
    \item $w'$ immediately precedes $w$: $w' \prec w$, and no other
    \textsl{writes} are between $w$ and $w'$
    \item $r'_{ft} \in [w_{st}, r_{st}]$
  \end{enumerate}
\end{definition}

The concurrency pattern itself is not sufficient for old-new
inversion.
Only when the \textsl{read/write} semantics in the last two requirements
of Definition~\ref{def:oni} is also satisfied, does an old-new inversion
arise.
Thus, we define the \mbox{\emph{read-write pattern}} conditioning on a
\emph{concurrency pattern} as follows.

\begin{definition} \label{def:read-write-pattern}
Given a concurrency pattern consisting of $r,r',w, \textrm{and } w'$, exactly as
those in \mbox{Definition~\ref{def:concurrency-pattern}}, the \textbf{read-write
pattern} requires
  \begin{enumerate}
    \setcounter{enumi}{3}
    \item $r = R(w')$
    \item $r' = R(w)$
  \end{enumerate}
\end{definition}

In this way, an ``old-new inversion'' occurs if and only if the read-write
pattern arises given that a corresponding concurrency pattern has emerged.
A concurrency pattern %(thus, an ``old-new inversion'')
may contain more than one such $r'$ defined in
Definition~\ref{def:concurrency-pattern}, as illustrated in
Figure~\ref{fig:old-new-inversion}.
Let $\mathrm{R'}$ be a random variable denoting the number of $r'$s in a
concurrency pattern.
Then, a \mbox{read-write} pattern arises if for \emph{some} $r'$,
Definition~\ref{def:read-write-pattern} is satisfied.
Therefore, the probability of \mbox{``old-new inversions''} conditioning on
$\mathrm{R' = m}$ ($m \geq 1$; $m$ can be as large as the number of all
\textsl{read} operations) is the product of the probability of the concurrency
patterns conditioning on $\mathrm{R' = m}$ and the probability of the read-write
patterns conditioning on $\mathrm{R' = m}$.
By the law of total probability, we obtain
\begin{equation} \label{equ:oni-factors}
  \begin{split}
 \mathbb{P} &\{ \text{violation of atomicity} \} = \mathbb{P} \{ \mathrm{ONI} \}
 \\
 &= \sum_{m \geq 1} \mathbb{P} \{ \mathrm{ONI \mid R' = m} \}
 \\
 &= \sum_{m \geq 1} \mathbb{P} \{ \mathrm{CP \mid R' = m} \}
 \times \mathbb{P} \{ \mathrm{RWP \mid R' = m} \}.
 \end{split}
\end{equation}

In the following two subsections, we propose a stochastic queueing model and a
timed balls-into-bins model to analyze the concurrency pattern and read-write
pattern in Equation~(\ref{equ:oni-factors}), respectively.
The frequently used notations and formulas are summarized in
Table~\ref{tbl:notation-formula}.
%%%%%%%%%%%%%%%%%%% notation-formula table %%%%%%%%%%%%%%%%
\begin{table*}[t!]
  \renewcommand{\arraystretch}{1.2}
  \caption{Notations and formulas.}
  \label{tbl:notation-formula}
  \centering
  \begin{tabular}{|cccc|}
    \hline
    \multicolumn{2}{|c|}{
      $N$: number of clients	\quad
      $n$: number of replicas  \quad
      $q \triangleq \lfloor n/2 \rfloor + 1$
    }&
    \multicolumn{2}{|c|}{
    Beta function: $B(x,y) = \int_{0}^{1} t^{x-1} (1-t)^{y-1} dt$
    }
    \\ \hline
    \multicolumn{2}{|c|}{
      $\lambda$: issue/arrival rate of operations,
      $\mu$: service rate of operations
    }&
    \multicolumn{2}{|c|}{
    $r \triangleq \frac{(2 \lambda + \mu)^2}{2(\mu + \lambda)^2}, \quad
    s \triangleq \frac{1}{2} \frac{\mu}{\mu + \lambda}, \quad
    p_0 \triangleq \frac{1}{2} \left( 1 + (\frac{\lambda}{\mu + \lambda})^2
    \right)$
    }
    \\ \hline
    \multicolumn{2}{|c|}{
      $\lambda_{r}$: rate for \textsl{read} latency,
      $\lambda_{w}$: rate for \textsl{write} latency,
      $\alpha = \frac{\lambda_{r}}{\lambda_{w} + \lambda_{r}}$
    }&
    \multicolumn{2}{|c|}{
%       Beta function: $B(x,y) = \int_{0}^{1} t^{x-1} (1-t)^{y-1} dt$
	$t = \frac{1}{\lambda}, \qquad t' = \frac{2\lambda - \mu}{2 \lambda \mu}$
    }
    \\ \hline
    \multicolumn{4}{|c|}{\small \hspace{-8pt}
      $\begin{aligned}
		  &J_1 = \lambda_{r} \int_{0}^{t'} e^{-\lambda_{r} (n-q+1) s}
		  \left(1-e^{-\lambda_{r} s}\right)^{q-1} ds
		  \\
		  &+ \sum_{k=0}^{n-q} \frac{\binom{q-1}{k-1}
		  \binom{n-q}{n-q-k}}{\binom{n}{n-q}} \lambda_{r}^{q} \, e^{\lambda_{w} t'}
		  \int_{t'}^{\infty} e^{-(\lambda_{w} + \lambda_{r}) s}
		  \left( \frac{1-e^{-\lambda_{r} t'}}{\lambda_{r}} + e^{\lambda_{w} t'}
		  \frac{e^{-(\lambda_{w} + \lambda_{r}) t'} - e^{-(\lambda_{w} + \lambda_{r})
		  s}}{\lambda_{w} + \lambda_{r}} \right)^{k-1}
		  \left(\frac{1 - e^{-\lambda_{r} s}}{\lambda_{r}}\right)^{q-k}
		  e^{-\lambda_{r}(n-q)s}ds
		  \\
		  &+ \sum_{k=0}^{n-q} \frac{\binom{q-1}{k} \binom{n-q}{n-q-k}}{
		  \binom{n}{n-q}} \lambda_{r}^{q} \, \int_{t'}^{\infty} e^{-\lambda_{r} s}
		  \left( \frac{1-e^{-\lambda_{r} t'}}{\lambda_{r}} + e^{\lambda_{w} t'}
		  \frac{e^{-(\lambda_{w} + \lambda_{r}) t'} - e^{-(\lambda_{w} +
		  \lambda_{r}) s}}{\lambda_{w} + \lambda_{r}} \right)^{k}
		  \left(\frac{1 - e^{-\lambda_{r} s}}{\lambda_{r}}\right)^{q-1-k}
		  e^{-\lambda_{r}(n-q)s} ds.
	  \end{aligned}$
    }
    \\ \hline
\end{tabular}
\end{table*}
%%%%%%%%%%%%%%%%
%%%%%%%%%%%%%%%%%%%
\subsection{Quantifying the Rate of Concurrency Patterns}
\label{subsection:quantifying-concurrency-pattern}

To quantify the rate of concurrency patterns conditioning on $\mathrm{R' = m}$, we need
an analytical model of the \emph{workload} consisting of one sequence of
\textsl{read/write} operations from each client.
For each client, the characteristics of its workload are captured by the rate of
operations issued by it and the service time of each operation (i.e.,
$[o_{st}, o_{ft}]$).
We assume a Poisson process with parameter $\lambda$ for the former one and an
exponential distribution with parameter $\mu$ for the latter one.
The scenario of each client issuing a sequence of \textsl{read/write} operations
is then encoded into a queueing model.

We thus consider $N$ independent, parallel $M/M/1$ queues (i.e., a
single-server exponential queueing system), all with arrival rate
$\lambda$ and service rate $\mu$ \cite{Ross10}.
% Notice that the single server here enforces well-formed executions.
For each $M/M/1$ queue, we use the ``first come first served'' discipline and
assume for simplicity that, if there is any operation in service, no more
operations can enter it.
The queue $Q_0$ represents the single writer.

To compute the probability that a concurrency pattern occurs in such a queueing
system in the long run, we go through the following three steps.

\emph{Step 1: What is the stationary distribution for any two queues?}

Let $X^{i}(t)$ be the number of operations in queue $i$ at time $t$.
Then $X^{i}(t)$ is a continuous-time Markov chain with only two states:
$0$ when the queue is empty and $1$ when some operation is being served.
Its stationary distribution is:
% {\small
% \begin{displaymath}
% P_0 \triangleq P \left(X^{i}(\infty) = 0 \right) = \frac{\mu}{\mu +
% \lambda},\quad
% P_1 \triangleq P \left(X^{i}(\infty) = 1 \right) = \frac{\lambda}{\mu + \lambda}
% \end{displaymath}
% }
\begin{align*}
  P_0 &\triangleq P \left(X^{i}(\infty) = 0 \right) = \frac{\mu}{\mu + \lambda}
  \\
  P_1 &\triangleq P \left(X^{i}(\infty) = 1 \right) = \frac{\lambda}{\mu +
  \lambda}
\end{align*}
Let $Y(t) = \left(X^{i}(t), Y^{j}(t) \right)$ be the vector of the numbers of
operations in queues $Q_i$ and $Q_j$. Since any two queues are independent, $Y(t)$ is a
continuous-time Markov chain with four states $(0,0), (0,1), (1,0), \text{ and } (1,1)$.
Its stationary distribution is:
\begin{gather*}
  P_{0,0} = \frac{\mu}{\mu + \lambda} \frac{\mu}{\mu + \lambda} =
  \frac{\mu^{2}}{(\mu + \lambda)^{2}}	\\
  P_{0,1} = P_{1,0} = \frac{\mu \lambda}{(\mu + \lambda)^{2}} \quad P_{1,1} =
  \frac{\lambda^{2}}{(\mu + \lambda)^{2}}
\end{gather*}
where,
\begin{displaymath}
  P_{i,j} \triangleq P \left( Y\left(\infty\right) = \left(i,j\right) \right)
  \qquad i, j \in \{ 0, 1 \}.
\end{displaymath}
% $$P(Y(\infty) = (0,0)) = \frac{\mu}{\mu + \lambda} \frac{\mu}{\mu + \lambda}
% \triangleq P_{0,0}$$
% $$P(Y(\infty) = (0,1)) = \frac{\mu}{\mu + \lambda} \frac{\lambda}{\mu + \lambda}
% \triangleq P_{0,1}$$
% $$P(Y(\infty) = (1,0)) = \frac{\lambda}{\mu + \lambda} \frac{\mu}{\mu + \lambda}
% \triangleq P_{1,0}$$
% $$P(Y(\infty) = (1,1)) = \frac{\lambda}{\mu + \lambda}
% \frac{\lambda}{\mu + \lambda} \triangleq P_{1,1}$$

\emph{Step 2: Given a \textsl{read} $r$ in $Q_i$, what is the probability of the
event, denoted $E$, that it starts during the service period of some
\textsl{write} $w$ in $Q_0$ (formally, $r_{st} \in [w_{st}, w_{ft}]$ in
Definition~\ref{def:concurrency-pattern})?}

The probability of $E$ equals the probability that when $r$ arrives at $Q_i$, it
finds $Q_i$ empty (denoted $E_i$) and as a bystander $Q_0$ full (denoted $E_0$).
Since events $E_i$ and $E_0$ are independent, we have
\begin{align*}
  P(E) &= P(E_i \land E_0) = P(E_i) \cdot P(E_0) \\
  &= P_0 \cdot P_1 \qquad {\text{(by the PASTA property} \cite{Ross10})} \\
  &= \frac{\mu \lambda}{(\mu + \lambda)^{2}}
\end{align*}

% \emph{Step 3: Conditioning on Step 2, what is the distribution of
% the service-starting time lag $L = r_{st} - w_{st}$?}
%
% The service-starting time lag $L = r_{st} - w_{st}$ is exactly the
% inter-arrival time of $Q_i$, which is exponential with rate $\lambda$.

\emph{Step 3: Conditioning on Step 2, what is the probability of the event,
denoted $E_{N-1,m}$, that there are totally $m$ \textsl{read} operations
(denoted $r'$) in $N-1$ queues (besides $Q_0$) which finish during the time
period $[w_{st}, r_{st}]$ (formally, $r'_{ft} \in [w_{st}, r_{st}]$ in
Definition~\ref{def:concurrency-pattern})?}

First, the length $L = r_{st} - w_{st}$ of the time period $[w_{st}, r_{st}]$
is exactly the inter-arrival time of $Q_i$, which is exponential with rate
$\lambda$.

The calculations in Appendix \ref{appendix:calculation} yield
\begin{align}  \label{equ:concurrency-pattern}
  \mathbb{P} &\{ \mathrm{CP \mid R' = m} \} = \mathbb{P}(E_{N-1,m})
   \nonumber\\
   &= \sum_{k=0}^{N-2} \binom{N-1}{k} \binom{m-1}{N-k-2} p_0^{k}
   r^{N-k-1} s^{m},
\end{align}
when $m \geq 1$.
For the special case $m = 0$, we have
\begin{displaymath}
  \mathbb{P} \{ \mathrm{CP \mid R' = 0} \} = \mathbb{P}(E_{N-1,0}) = p_0^{N-1}.
\end{displaymath}
% where,
% \begin{displaymath}
%   r \triangleq \frac{(2 \lambda + \mu)^2}{2(\mu + \lambda)^2},
%   s \triangleq \frac{1}{2} \frac{\mu}{\mu + \lambda},
%   p_0 \triangleq \frac{1}{2} \left( 1 + (\frac{\lambda}{\mu + \lambda})^2
%   \right).
% \end{displaymath}
Summing over $m$ ($m \geq 1$), we also get the probability that there exists a
concurrency pattern (for some \textsl{read} $r$):
\begin{equation}  \label{equ:concurrency-pattern-summation}
  \mathbb{P} \{ \mathrm{CP} \} = 1 - \mathbb{P} \{ \mathrm{CP \mid R' = 0} \} =
  1 - p_0^{N-1}
\end{equation}

\subsection{Quantifying the Rate of Read-Write Patterns}
\label{subsection:quantifying-read-write-pattern}

%%%%%%%%% concurrency pattern %%%%%%%%%
\begin{figure*}[th]
  \centering
  \includegraphics[width = 0.85\textwidth]
  {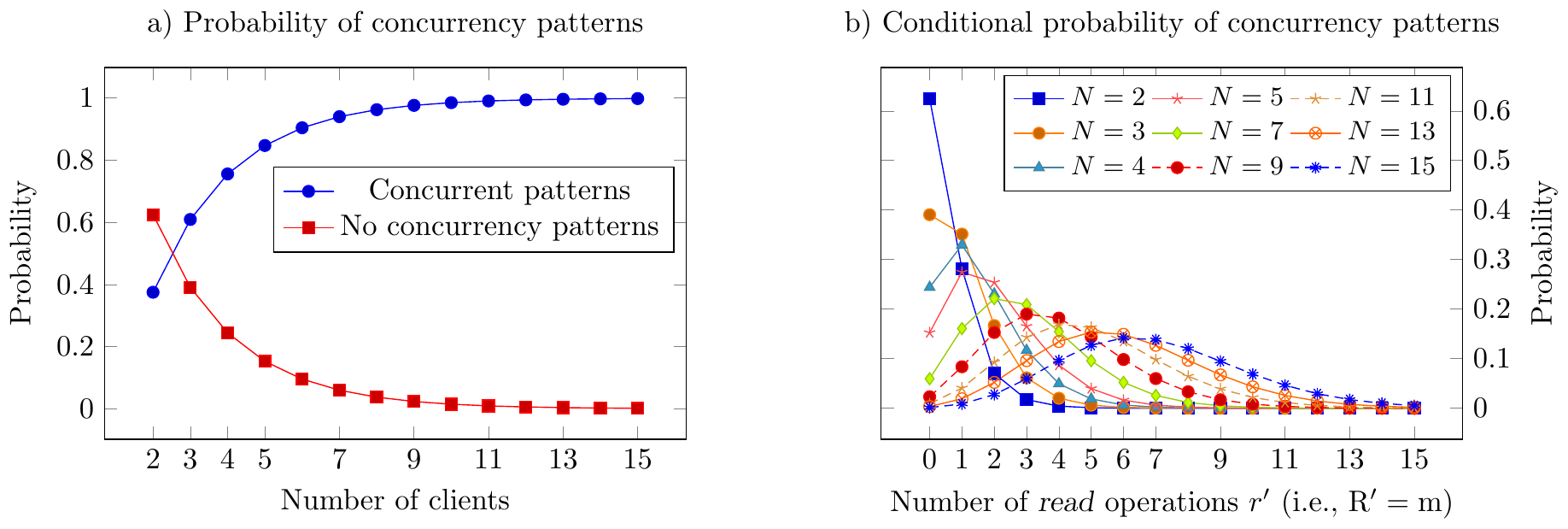}
  \caption[labelInTOC]{The probability of concurrency patterns: a) with vs.
  without concurrency patterns; b) conditioning on $\mathrm{R' = m}$
  ($\lambda = 10 s^{-1}, \mu = 10 s^{-1}$).}
  \label{fig:cp-plot}
\end{figure*}
%%%%%%%%% concurrency pattern %%%%%%%%%

Given the concurrency patterns, we further quantify the rate of read-write
patterns conditioning on \mbox{$\mathrm{R' = m}$}:
\[
  r = R(w') \land \exists r': r' = R(w)
\]
where $r'$ is among the $m$ \textsl{read} operations in \emph{Step 3} in
Section~\ref{subsection:quantifying-concurrency-pattern}.
To this end, we shall explore in detail the majority quorum systems used in the
2AM algorithm.
We assume that \itshape 1\upshape) no node failure or link failure occurs; and
\itshape 2\upshape) to complete an operation (\textsl{read} or \textsl{write}),
the client accesses all the $n$ replicas and wait for the first $q \triangleq
\lfloor n/2 \rfloor + 1$ acknowledgments from them.
It follows that:
\begin{align} \label{eqn:rwp-in-text}
  \mathbb{P} &\{ \mathrm{RWP \mid R' = m} \}
  \nonumber\\
  &= \mathbb{P} \{ r = R(w') \land \exists r': r' = R(w) \}
  \nonumber\\
  &\leq \mathbb{P} \{ r \neq R(w) \land \exists r': r' = R(w) \}
  \nonumber\\
  &= \mathbb{P} \{ r \neq R(w) \} \times \mathbb{P} \{ \exists r': r' = R(w)
  \mid r \neq R(w) \}
  \\
  &= \mathbb{P} \{ r \neq R(w) \} \times \Big(1 - \mathbb{P} \{ r' \neq R(w)
  \mid r \neq R(w) \}^{m}\Big)  \nonumber
\end{align}
\begin{comment}
\begin{align}
  \mathbb{P} &\{ \mathrm{RWP \mid R' = m} \}
  \nonumber\\
  &= \mathbb{P} \{ r = R(w') \land \exists r': r' = R(w) \}
  \nonumber\\
  &= \mathbb{P} \{ r = R(w') \} \times \mathbb{P} \{ \exists r': r' = R(w) \mid
  r = R(w') \}
  \nonumber\\
  &= \mathbb{P} \{ r = R(w') \} \times \big(1 - \mathbb{P} \{ r' \neq R(w) \mid
  r = R(w') \}^{m} \big)
  \nonumber\\
  &= \mathbb{P} \{ r = R(w') \}
  \nonumber\\
  &\; \times \Big( 1 - \big(\mathbb{P} \{r' \neq R(w) \} - \mathbb{P} \{ r'
  \neq R(w) \mid r \neq R(w') \}\big)^{m} \Big)
  \nonumber\\
  &\leq \mathbb{P} \{ r \neq R(w) \}
  \nonumber\\
  &\; \times \Big( 1 - \big(\mathbb{P} \{r' \neq R(w) \} - \mathbb{P} \{ r'
  \neq R(w) \mid r \neq R(w') \}\big)^{m} \Big)
\end{align}
\end{comment}
where $r \neq R(w)$ (resp. $r' \neq R(w)$) denotes that $r$ (resp. $r'$) does
not read from $w$.
The inequality is due to the fact that $r = R(w')$ implies $r \neq R(w)$.
% The justification of the inequality of $\mathbb{P} \{ r = R(w') \} \leq
% \mathbb{P} \{ r \neq R(w) \}$ \textcolor{red}{in the last step} is given in
% Appendix~\ref{appendix:rwp-upper-bound}.
We then focus on the calculations of $\mathbb{P} \{ r \neq R(w)
\}$ and $\mathbb{P} \{ r' \neq R(w) \mid r \neq R(w) \}$.
% requiring some $\textsl{read}$ should not read from some \textsl{write}.

Which \textsl{write} would be read from by some \textsl{read} depends on the
states of the replicas from which it collects the first $\lfloor n/2 \rfloor +
1$ acknowledgments.
% It depends on the states of the replicas accessed that which data version (and
% the corresponding \textsl{write}) is returned by a \textsl{read}.
The states of the replicas further depend on the timing issues in the 2AM
algorithm, such as message delays and the time lag between the events that the
messages are sent.
Taking into account the timing issues, we propose the timed balls-into-bins
model for the \textsl{read} and \textsl{write} procedures in the 2AM algorithm.
Let $D_{r}$ (resp. $D_{w}$) be a (non-negative) continuous random variable
denoting the message delay for \textsl{read} (resp. \textsl{write}) operations
during a communication round-trip.
Let $T$ be a (non-negative) continuous random variable denoting the time lag
between the time when two messages of interest are sent, and $t$ a
realization (or called an observed value) of $T$.

In the \emph{timed balls-into-bins model}, there are $n$ bins (corresponding to
$n$ replicas).
Consider two robots $R_1$ and $R_2$ (corresponding to \textsl{read}
or \textsl{write} operations) which can produce multiple balls (corresponding
to messages) instantaneously.
At time 0, robot $R_1$ \itshape 1\upshape) produces $n$ balls instantaneously;
\itshape 2\upshape) Immediately these $n$ balls are independently sent to the
$n$ bins, one ball per bin;
\itshape 3\upshape) The delays for the balls going from the robot to its
destination bin are independent and identically distributed with the same
distribution as $D_{r}$ or $D_{w}$ as defined above, depending on whether the
robot represents a \textsl{read} or a \textsl{write}.

At time $t$ (defined above), robot $R_2$ independently does exactly the same
thing as robot $R_1$ does (i.e., \itshape 1\upshape), \itshape 2\upshape), and
\itshape 3\upshape) for robot $R_1$ above).

Each probability to calculate is related to an event in an instantiation of the
timed balls-into-bins model.

To calculate $\mathbb{P} \{ r \neq R(w) \}$, we are concerned with the model in
which the robots $R_1$ and $R_2$ represent the \textsl{write} operation $w'$ and
the \textsl{read} operation $r$ involved in an ``old-new inversion'', respectively.
Furthermore, we assume that the random variable $D_{r}$ (resp. $D_{w}$) for time
delay is exponentially distributed with rate $\lambda_r$ (resp. $\lambda_w$).
The time lag $T$ between the events that $w'$ and $r$ are issued
(meanwhile messages are sent to replicas) corresponds to the time period
$[w_{st}, r_{st}]$.
That is to say, $T$ is an exponential random variable with rate $\lambda$, as
shown in Section~\ref{subsection:quantifying-concurrency-pattern} (See
\emph{Step 3}).
For simplicity, we take the time lag $t$ to be the expectation of $T$, i.e., $t
= \frac{1}{\lambda}$.
% Finally, the case of $r \neq R(w)$ occurs if and only if none of the $q
% \triangleq \lfloor n/2 \rfloor + 1$ acknowledgments $r$ collects from the
% replicas contains the value written by $w$.
Finally, we are interested in the time point $t'$ when exactly $q \triangleq
\lfloor n/2 \rfloor + 1$) of the $n$ bins have received the balls from $R_2$
(i.e., $r$), and denote the set of these $\lfloor n/2 \rfloor + 1$ bins by $B$.
In terms of the timed balls-into-bins model, the case of $r \neq R(w)$
corresponds to the event $E$ that none of the $q \triangleq \lfloor n/2 \rfloor
+ 1$ bins in $B$ receives a ball from $R_1$ (i.e., $w$) before it receives a
ball from $R_2$ (i.e., $r$).

The calculations in Appendix~\ref{appendix:rwp-r-w} yield
\begin{align} \label{eqn:rwp1-in-text}
  \mathbb{P} \{ r \neq R(w) \} = e^{-q \lambda_{w} t} \frac{\alpha^{q}
  B\left(q, \alpha (n-q) + 1 \right)} {B\left(q,n-q+1\right)},
\end{align}
where $\alpha = \frac{\lambda_{r}}{\lambda_{w} + \lambda_{r}}$ and $B$ denotes
the Beta function.

\begin{comment}
To calculate $\mathbb{P} \{ r' \neq R(w) \}$, we consider the timed
balls-into-bins model in which $R_1$ and $R_2$ represent $r'$ and $w'$,
respectively (see Appendix~\ref{appendix:rwp-rprimew-model}).
The calculations in Appendix~\ref{appendix:rwp-rprime-w} yield
\begin{align}
  \mathbb{P} \{ r' \neq R(w) \} = \frac{J_1}{B(q,n-q+1)},
\end{align}
where $J_1$ is given in Equation~(\ref{eqn:j1}) in
Appendix~\ref{appendix:rwp-rprime-w}.
\end{comment}

To calculate $\mathbb{P} \{ r' \neq R(w) \mid r \neq R(w') \}$, we
introduce a slightly generalized timed balls-into-bins model.
In the new model, robot $R_2$ picks $p$ ($0 < p \leq n$) bins uniformly at
random (without replacement) and sends a ball to each of them (see
Appendix~\ref{appendix:rwp-generalized-model}).
The calculations in Appendix~\ref{appendix:rwp-rprime-w} yield
\begin{align}	\label{eqn:rwp2-in-text}
  \mathbb{P}\{ r' \neq R(w) \mid r \neq R(w) \} = \left\{
  \begin{array}{ll}
    \frac{J_1}{B(q,n-q+1)} & \textrm{if } n > 2, \\
    1 & \textrm{if } n = 2.
  \end{array}\right.
\end{align}
% where $J_1$ is given in Equation~(\ref{eqn:j1}).

Substituting Equations~(\ref{eqn:rwp1-in-text}) and (\ref{eqn:rwp2-in-text})
into Equation~(\ref{eqn:rwp-in-text}) gives, for $n > 2$, the rate of read-write
patterns conditioning on $\mathrm{R' = m}$:
\begin{align} \label{equ:read-write-pattern}
  \mathbb{P} &\{ \mathrm{RWP \mid R' = m} \}
  \nonumber\\
  &\leq \mathbb{P} \{ r \neq R(w) \} \times \Big(1 - \mathbb{P} \{ r' \neq R(w)
  \mid r \neq R(w) \}^{m}\Big)
  \nonumber\\
  &\le e^{-q \lambda_{w} t} \frac{\alpha^{q} B\left(q, \alpha (n-q) + 1 \right)}
  {B\left(q,n-q+1\right)}
  \nonumber\\
  &\quad \cdot \left(1- \left(\frac{J_1}{B(q,n-q+1)}\right)^{m}
  \right).
\end{align}
% where $\alpha = \frac{\lambda_{r}}{\lambda_{w} + \lambda_{r}}$, $B$ denotes
% the Beta function, and $J_1$ is given in Equation~(\ref{eqn:j1}).
For $n=2$, we have
\begin{displaymath}
  \mathbb{P} \set{\mathrm{RWP \mid R' = m}} = 0.
\end{displaymath}
Notice that $\mathbb{P} \{ \mathrm{RWP \mid R' = 0} \} = 0$ since there are no
concurrency patterns at all.
%%%%%%%%%%%%%%%%%%%%%%%%
\subsection{Numerical Results and Discussions} \label{subsection:analysis}

In light of the complicated analytical formulation, we present the numerical
results on \emph{concurrency patterns, read-write patterns, and old-new
inversions}.
The numerical results have not only demonstrated that ``old-new inversions''
(and thus, atomicity violations) rarely occur as expected, but also clearly
revealed that the read-write patterns dominate in guaranteeing such rare
violations.

Figure~\ref{fig:cp-plot} presents the probability of concurrency patterns, given
$\lambda = 10 s^{-1} \textrm{ and } \mu = 10 s^{-1}$, meaning that the
expected arrival rate is 10 operations per second and the expected service time
is 100 ms.
First of all, Figure~\ref{fig:cp-plot}a) shows that the probability of
concurrency pattern is quite high, and it rapidly increases with the number of
clients.
For example when $N=15$, it nearly reaches $1$:
intuitively, for each \textsl{read} $r$, there almost always exist concurrency
patterns involving it.
Figure~\ref{fig:cp-plot}b) further explores the probability of concurrency
patterns conditioning on the number $m$ of \textsl{reads} $r'$ (i.e.,
$\mathbb{P} \{\mathrm{CP \mid R' = m} \}$).
Here $m=0$ indicates that there are no concurrency patterns at all,
corresponding to the (square-marked) line at the bottom in
Figure~\ref{fig:cp-plot}a).
One key observation from Figure~\ref{fig:cp-plot}b) is that the
conditional probability of concurrency patterns concentrates on
the small values of $m$'s, and for each $N$ the value of $m$ which achieves
the maximum is smaller than $N$.
This observation partly justifies the assumption made in the model for
calculating $\mathbb{P}\set{r' \neq R(w) \mid r \neq R(w)}$ (see
Appendix~\ref{appendix:rwp-generalized-model}) that there is at most one such
$r'$ in a single (client) process.

% For a fixed value of $m$ ($m \geq 1$), the more replicas involved, the more
% operations issued and the higher proportion of concurrency patterns containing
% exactly $m$ $r'$s.
% % the bigger replication factor means the more operations and leads to the higher
% % proportion of concurrency patterns which contain exactly $m$ $r'$s
% % (Figure~\ref{fig:patterns-analysis}a).
% \textcolor{red}{On the other hand}, for a fixed number of replicas (and thus the
% workload is also fixed), the bigger $m$ is, the lower the proportion of
% concurrency patterns is.
% In this way, the parameter $m$ can be seen as an indicator of how
% \emph{difficult} it is to form the concurrency pattern.

Figure~\ref{fig:rwp-plot}, as well as Table~\ref{table:rwp-probabilities},
presents the probability of read-write patterns,
given $\lambda = 10 s^{-1}, \mu = 10 s^{-1}, \textrm{and } \lambda_{r} =
\lambda{w} = 20 s^{-1}$.
The latter two parameters mean that the expected message delay is 50 ms.
According to Equation~(\ref{equ:read-write-pattern}), we distinguish the
probability $\mathbb{P} \{ r \neq R(w) \}$ from another one $1 - \mathbb{P} \{
r' \neq R(w) \mid r \neq R(w) \}^{m}$ (in Figure~\ref{fig:rwp-plot}, we take the
extreme value of $m=1$), and observe that the former dominates in keeping the
probability of read-write patterns quite low.
The observation that $\mathbb{P} \{ r \neq R(w) \} $ is quite low
has demonstrated the effectiveness of the majority quorum system used in the 2AM
algorithm, under which a \textsl{read} would, with a high probability,
not miss a concurrent \textsl{write} that starts earlier.
In addition, \emph{if} a \textsl{read} $r$ has happened to miss such a
concurrent \textsl{write}, it is still quite likely to avoid an old-new
inversion: $r$ can reasonably infer, from the low values of $1 - \mathbb{P} \{
r' \neq R(w) \mid r \neq R(w) \}$, that the preceding \textsl{reads} $r'$ would
not have read from that \textsl{write} either.

% table for r != R(w) and 1 - r' != R(w) | r != R(w)
\newcommand*{\thead}[1]{\multicolumn{1}{|c|}{\bfseries #1}}
\begin{table*}[t]
  \renewcommand{\arraystretch}{1.2}
  \caption{Numerical results on the probabilities of $\{r \neq R(w)\}$
  and $1 - \{r' \neq R(w) \mid r \neq R(w)\}$.}
  \label{table:rwp-probabilities}
  \centering
  \begin{tabular}{|>{\boldmath$}c<{$}|>{$}l<{$}|>{$}c<{$}||>{\boldmath$}c<{$}|>{$}c<{$}|>{$}c<{$}|}
    \hline
    \thead{\text{\bfseries \# replicas}} &
%     \mathbb{P}\set{r \neq R(w)} &
    \multicolumn{1}{c|}{$\mathbb{P}\set{r \neq R(w)}$} &
%     1 - \mathbb{P}\set{r' \neq R(w) \mid r \neq R(w)} &
	\begin{tabular}[c]{@{}r@{}}$1 - \mathbb{P} \{r' \neq R(w)$ \\ $\mid r \neq
	R(w)\}$ \end{tabular} &
	\thead{\text{\bfseries \# replicas}} &
    \mathbb{P}\set{r \neq R(w)} &
%     1 - \mathbb{P}\set{r' \neq R(w) \mid r \neq R(w)}
	\begin{tabular}[c]{@{}r@{}}$1 - \mathbb{P} \{r' \neq R(w)$ \\ $\mid r \neq
	R(w)\}$ \end{tabular}
    \\ \hline \hline
    2 		& 0.00457891				& 1.0
    & 9 	& 8.51249 \times 10^{-6}	& 0.0243758
    \\ \hline
    3 		& 0.00732626 				& 0.0409628
    & 10 	& 7.20025 \times 10^{-7}	& 0.0353241
    \\ \hline
 	4 		& 0.000566572				& 0.0561367
 	& 11 	& 8.89660 \times 10^{-7}		& 0.0203645
    \\ \hline
    5 		& 0.00077461  				& 0.0356626
 	& 12 	& 7.60436 \times 10^{-8}	& 0.0294186
 	\\ \hline
 	6 		& 0.0000628992				& 0.0511399
 	& 13 	& 9.28973 \times 10^{-8}	& 0.0171705
 	\\ \hline
 	7 		& 0.0000813243				& 0.0294467
 	& 14 	& 8.00055 \times 10^{-9}	& 0.0246974
 	\\ \hline
    8 		& 6.77295 \times 10^{-6}	& 0.0426608
    & 15 	& 9.69478 \times 10^{-9}	& 0.0145951
    \\ \hline
  \end{tabular}
\end{table*}

% table for cp, rwp, and oni
\begin{table*}[!t]
  \renewcommand{\arraystretch}{1.2}
  \caption{Numerical results on the probabilities of
  concurrency patterns, read-write patterns, and old-new inversions.}
  \label{table:oni-probabilities}
  \centering
  \begin{tabular}{|>{\boldmath$}c<{$}|>{$}l<{$}|>{$}l<{$}|>{$}l<{$}||>{\boldmath$}c<{$}|>{$}l<{$}|>{$}l<{$}|>{$}l<{$}|}
    \hline
    \thead{\text{\small \bfseries \# replicas}} &
    \multicolumn{1}{|c|}{$\mathbb{P}\set{\mathrm{CP}}$} &
    \multicolumn{1}{|c|}{$\mathbb{P}\set{\mathrm{RWP} \mid \mathrm{CP}}$} &
    \multicolumn{1}{|c||}{$\mathbb{P}\set{\mathrm{ONI}}$} &
    \thead{\text{\small \bfseries \# replicas}} &
    \multicolumn{1}{|c|}{$\mathbb{P}\set{\mathrm{CP}}$} &
    \multicolumn{1}{|c|}{$\mathbb{P}\set{\mathrm{RWP} \mid \mathrm{CP}}$} &
    \multicolumn{1}{|c|}{$\mathbb{P}\set{\mathrm{ONI}}$}
    \\ \hline \hline
    2 		& 0.28125		 	& 0.		  	& 0.
    & 9 	& 0.9447 			& 7.06025 \times 10^{-6} 	& 7.30744 \times 10^{-7}
    \\ \hline
    3 		& 0.518555 			& 0.00088802 	& 0.000203683
    & 10 	& 0.95874 			& 1.04312 \times 10^{-6} 	& 9.93356 \times 10^{-8}
    \\ \hline
 	4 		& 0.677307 			& 0.000183791 	& 0.0000352958
 	& 11 	& 0.968604 			& 9.37995 \times 10^{-7} 	& 8.16935 \times 10^{-8}
    \\ \hline
    5 		& 0.781222 			& 0.000266569 	& 0.0000437181
 	& 12 	& 0.975675 			& 1.34085 \times 10^{-7} 	& 1.08822 \times 10^{-8}
 	\\ \hline
 	6 		& 0.849318 			& 0.0000450835 	& 6.49226 \times 10^{-6}
 	& 13 	& 0.98085 			& 1.16911 \times 10^{-7} 	& 8.77158 \times 10^{-9}
 	\\ \hline
 	7 		& 0.89429 			& 0.0000478926 	& 6.08721 \times 10^{-6}
 	& 14 	& 0.984717 			& 1.63195 \times 10^{-8} 	& 1.15178 \times 10^{-9}
 	\\ \hline
    8 		& 0.924335 			& 7.43561 \times 10^{-6} 	& 8.53810 \times 10^{-7}
    & 15 	& 0.987662 			& 1.39573 \times 10^{-8} 	& 9.18283 \times 10^{-10}
    \\ \hline
  \end{tabular}
\end{table*}

\begin{figure}[!t]
  \centering
  \includegraphics[width = 0.38\textwidth]{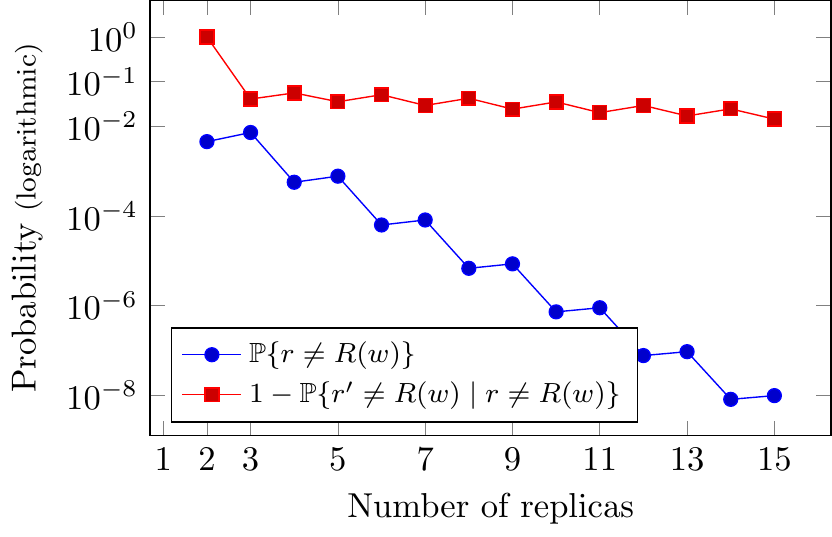}
%   {fig/vldb-rwp-pattern-plot-2.pdf}
  \caption[labelInTOC]{The probability of read-write patterns ($\lambda = 10
  s^{-1}, \mu = 10 s^{-1}, \lambda_{r} = \lambda{w} = 20 s^{-1}$).}
  \label{fig:rwp-plot}
\end{figure}

Substituting Equations~(\ref{equ:concurrency-pattern}) and
~(\ref{equ:read-write-pattern}) into Equation (\ref{equ:oni-factors}), we obtain
the rate of violating atomicity:
\begin{align}  \label{eqn:oni}
  \mathbb{P} &\{ \text{violation of atomicity} \} = \mathbb{P} \{\mathrm{ONI}\}
  \nonumber\\
  &= \sum_{m \geq 1} \mathbb{P} \{ \mathrm{CP \mid R' = m} \}
  \times \mathbb{P} \{ \mathrm{RWP \mid R' = m} \}
  \nonumber\\
  &\approx \left( \sum_{k=0}^{N-2} \binom{N-1}{k} \binom{m-1}{N-k-2} p_0^{k}
  r^{N-k-1} s^{m} \right) \nonumber\\
  &\quad \cdot e^{-q \lambda_{w} t} \frac{\alpha^{q} B\left(q, \alpha (n-q)
  + 1 \right)} {B\left(q,n-q+1\right)}
  \nonumber\\
  &\quad \cdot \left(1 - \left(\frac{J_1}{B(q,n-q+1)}\right)^{m}
  \right).
\end{align}
% where $N$ is the number of clients, $n$ is the number of replicas, $B$ denotes
% the Beta function, $J_1$ is given in Equation~(\ref{eqn:j1}), and
% \begin{gather*}
%   r \triangleq \frac{(2 \lambda + \mu)^2}{2(\mu + \lambda)^2},\,
%   s \triangleq \frac{1}{2} \frac{\mu}{\mu + \lambda},\,
%   t = \frac{1}{\lambda},\,
%   \\
%   p_0 \triangleq \frac{1}{2} \left( 1 + (\frac{\lambda}{\mu + \lambda})^2
%   \right),\,
%   q = \lfloor n/2 \rfloor + 1, \, \alpha = \frac{\lambda_{r}}{\lambda_{w} +
%   \lambda_{r}}.
% \end{gather*}

\begin{figure}[!t]
  \centering
  \includegraphics[width = 0.38\textwidth]{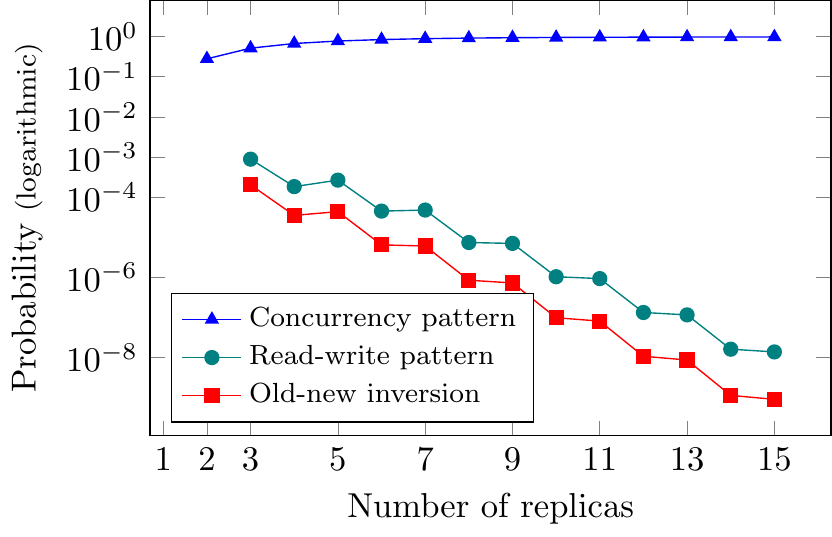}
%   {fig/vldb-cp-rwp-oni-plot.pdf}
  \caption[labelInTOC]{The probabilities of concurrency patterns, read-write
  patterns, and old-new inversions ($\lambda = 10 s^{-1}, \mu = 10 s^{-1},
  \lambda_{r} = \lambda{w} = 20 s^{-1}, N = n$).}
  \label{fig:oni-plot}
\end{figure}

Notice that Equation~(\ref{eqn:oni}) is an approximation since the timed
balls-into-bins model used for calculating the probability of read-write
patterns (specifically, for the case of $\set{r' \neq R(w) \mid r \neq R(w)}$ in
Appendix~\ref{appendix:rwp-generalized-model}) assumes that there is at most one
such $r'$ in a single process, while the model for calculating the probability
of concurrency patterns does not.

Figure~\ref{fig:oni-plot}, as well as Table~\ref{table:oni-probabilities},
presents the probability of old-new inversions according to
Equation~(\ref{eqn:oni}) with $N = n$.
We also list the probabilities of concurrency patterns and read-write patterns,
which are calculated as follows:
\begin{align*}
  \mathbb{P}\set{\mathrm{CP}} &= \sum_{m = 1}^{N-1} \mathbb{P}\set{\mathrm{CP}
  \mid \mathrm{R' = m}}
  \\
  \mathbb{P}\set{\mathrm{RWP} \mid \mathrm{CP}} &= \sum_{m = 1}^{N-1}
  \mathbb{P}\set{\mathrm{RWP} \mid \mathrm{R' = m}}.
\end{align*}
% Recall that (Equation~(\ref{equ:oni-factors}))
% \begin{align*}
%  \mathbb{P}\set{\mathrm{ONI}} &= \sum_{m \geq 1} \mathbb{P}\set{\mathrm{ONI
%  \mid R' = m}}
%  \\
%  &= \sum_{m \geq 1} \mathbb{P} \set{\mathrm{CP \mid R' = m}}
%  \times \mathbb{P} \set{\mathrm{RWP \mid R' = m}}.
% \end{align*}
% Notice that it does not necessarily hold that $\mathbb{P}\set{\mathrm{ONI}}
% = \mathbb{P}\set{\mathrm{CP}} \cdot \mathbb{P}\set{\mathrm{RWP} \mid
% \mathrm{CP}}$, as illustrated in Table~\ref{table:oni-probabilities}.

Based on Figure~\ref{fig:oni-plot} and Table~\ref{table:oni-probabilities}, we
first observe that the probability of old-new inversions (and thus, atomicity
violations) is sufficiently small, demonstrating that 2-atomicity and the 2AM
algorithm is ``good enough'' in distributed storage systems.
More importantly, it also reveals that the read-write patterns dominate in
guaranteeing such rare violations, compared to the concurrency patterns which
occur quite often.
% \emph{This insight $\ldots$}
% However, when $m$ is large enough ($m \geq 6$ here), the smaller replication factor
% leads to the higher proportion of read-write patterns in concurrency patterns
% (Figure \ref{fig:patterns-analysis}b).
% The key point here is, as $m$ increases, the probability of concurrency
% patterns drops sharply (Figure \ref{fig:patterns-analysis}a), while the
% probability of read-write patterns (in concurrency patterns) grows slowly
% (Figure \ref{fig:patterns-analysis}b).
% In this way, the rate of ``old-new inversions'' (and thus the violations of
% atomicity) is kept sufficiently small.
%%%%%%%%%%%%

Notice that the principles underlying our theoretical analysis (as well as the
numerical analysis) have been decoupled from the assumptions we adopt about the
networks and workloads.
These principles mainly consist of the introduction to old-new inversion, the
decomposition of it into concurrency pattern and read-write pattern, the
queueing model for analyzing concurrency patterns, and the timed balls-into-bins
model for analyzing read-write patterns.
Network conditions and workload types may vary in different scenarios.
However, the principles and the methodology of our analysis still apply.

% Table \ref{tbl:prob_of_patterns} shows the probability that there exists a
% concurrency pattern: $\mathbb{P} \{ \mathrm{CP} \} = 1 - \mathbb{P} \{
% \mathrm{CP \mid R' = 0} \} = 1 - p_0^{N-1}$.
% It increases with the replication factor.
% The analysis of concurrency pattern is validated by the experimental study
% in Section \ref{section:experiment}.
% And the analysis of read-write pattern is further refined there.

% Remember that an ``old-new inversion'' does not occur just
% because the presence of concurrency pattern of \textsl{read} and \textsl{write}
% operations;
% it requires those messages sent by clients to be received in some specific
% order by different replicas.
% How often that happens depends on how often operations (also with messages) are
% generated and on the nature of the communication medium.
% For example, if all the replicas are on the same Ethernet or wireless LAN, then
% it is very unlikely for messages to be received in different order by different
% servers, so old-new inversions will be extremely rare.
% In this sense, our analytical result can be regarded as a conservative upper
% bound on real measurement.
%%%%%%%%%%%%%%%%%%%%%%%%%%%%%%%%%%%%%%%%%%%%%%%%%
\section{Experiments and Evaluations} \label{section:experiment}

In this section, we empirically study the 2AM algorithm.
% under a variety of configurations and workload parameters.
To this end, we have implemented a prototype data storage system among mobile
phones, which provides \mbox{2-atomic} data access based on the 2AM algorithm
and atomic data access based on the ABD algorithm.
We compare the \textsl{read} latency in both algorithms.
We also measure the proportion of atomicity violations incurred in the
2AM algorithm.
%%%%%%%%%%%%%%%%%%%%%%%%
\subsection{Experimental Design}	\label{subsection:exp-design}

Our prototype system comprises a collection of Google Nexus5 smartphones
(CPU: Qualcomm Snapdragon\texttrademark\ 800, 2.26GHz, Memory: 16GB, Android:
4.4.2), equipped with 72Mbps wireless LAN.
In both algorithms, each phone acts as both a client and a server replica.
As a client, it collects its own execution trace for offline analysis.
Clocks on the phones are synchronized with the same desktop computer.

We explore three kinds of parameters:
\itshape 1\upshape) algorithm parameters: replication factor
(i.e., the number of phones) and consistency levels (i.e.,
atomicity or \mbox{2-atomicity});
\itshape 2\upshape) workload parameters: the number of
\textsl{read}/\textsl{write} operations issued by each client and the issue
rate on each client;
and \itshape 3\upshape) network parameter: the injected random delay in network
communication, modeling the various degrees of asynchrony.

We are concerned with two metrics:

%%%%%%%%%% experiment: latency %%%%%%%%%%%
\begin{figure*}[!t]
  \centering
    \includegraphics[width = 0.90\textwidth]{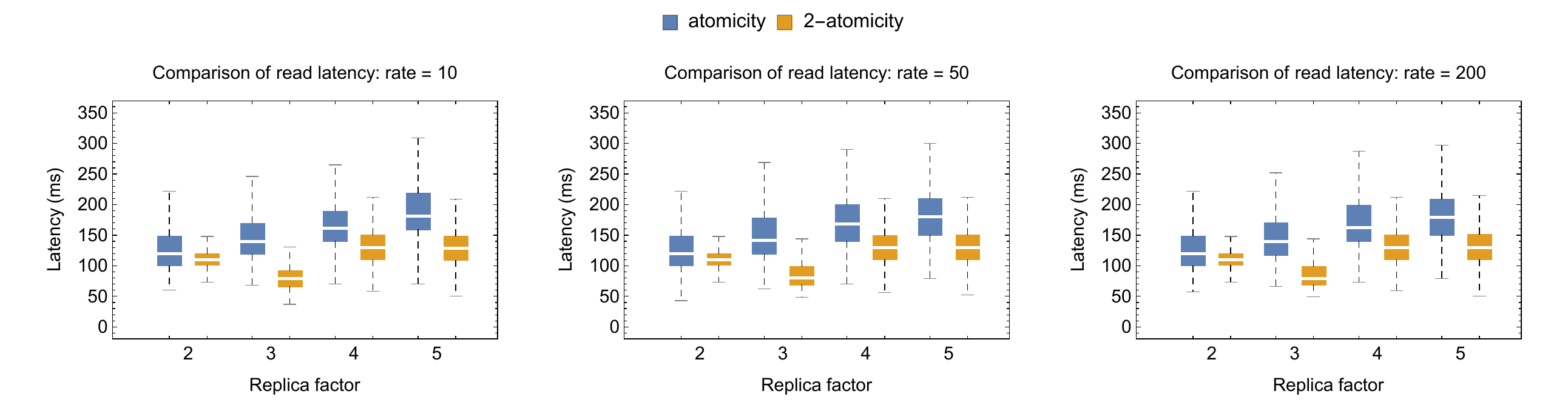}
    \caption[labelInTOC]{Comparison of \textsl{read} latency in both the ABD
    atomicity algorithm and the 2AM algorithm.}
    \label{fig:experiment-latency}
\end{figure*}
%%%%%%%%%% experiment: latency %%%%%%%%%%%

%%%%%%%%%% experiment: oni %%%%%%%%%%%
\begin{figure*}[th]
  \centering
    \includegraphics[width = 0.75\textwidth]
    {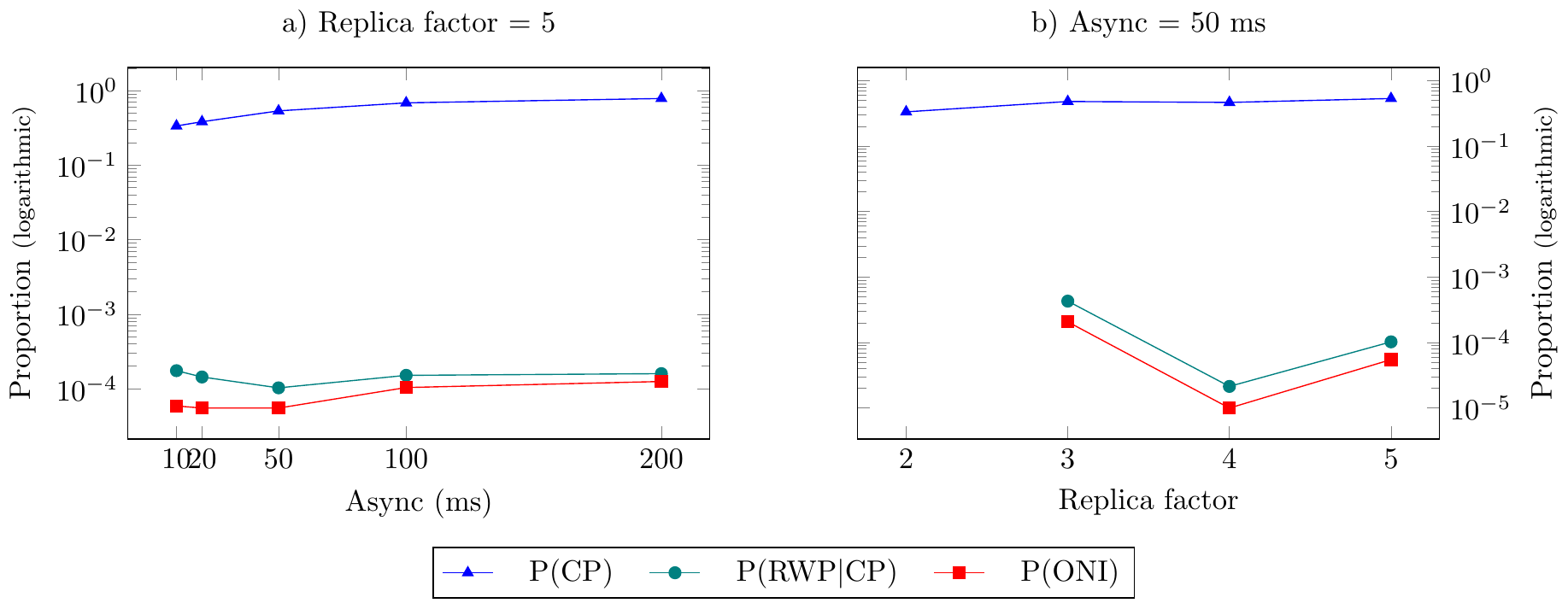}
%     {fig/cp-rwp-oni-prob.pdf}
    \caption[labelInTOC]{The proportions of concurrency patterns
    $\mathrm{P(CP)}$, read-write patterns among concurrency patterns $\mathrm{P(RWP|CP)}$, and
    ``old-new inversions'' $\mathrm{P(ONI)}$. (\emph{$\log 0$ is not defined and
    thus not shown.})}
    \label{fig:experiment-oni}
\end{figure*}
%%%%%%%%%% experiment: oni %%%%%%%%%%%

\emph{Latency:} We compare the \textsl{read} latency in both algorithms by
varying the replication factors and the issue rates of operations in the workload.
Each client issues \textsl{reads/writes} at a Poisson rate $\lambda$ (= 5, 10,
20, 50, 100, or 200) per second.
For each $\lambda$, the replication factors vary from 2 to 5.
Each reader issues 50, 000 \textsl{read} operations.
The single writer issues only \textsl{write} operations.
% In atomicity algorithm, the writer is required to issue 100, 000 operations to
% prevent it from terminating too much earlier than the other readers due to the
% fact that it only needs one round-trip.
In addition, the size of the keyspace is fixed to 1.
The key takes integer values from 0 to 4.

\emph{Violations of atomicity:} We quantify the violations of atomicity
incurred in the 2AM algorithm by varying the replication factors and the network
delays.
The replication factors vary from 2 to 5.
For each replication factor, the injected random delays in network
communication are uniformly distributed over integers in $[0, r)$ ($r$ can be
10, 20, 50, 100, and 200 ms).
% Notice that the average round-trip time in our wireless LAN is around 80ms.
Each client issues 200, 000 operations.
The single writer issues only \textsl{write} operations.
On each client, operations arrive at a Poisson rate of 50 per second so that the
system operates at its full capacity.
The size of the keyspace is 1 and the ``hotspot'' key takes integer values from
0 to 4.
% All the concerns aim to generate as many concurrency patterns and thus
% violations as possible.
%%%%%%%%%%%%%%%%%%%%%%%%
\subsection{Experimental Result 1: Latency}	\label{subsection:latency}

We visualize the latency data using box plots
(Figure~\ref{fig:experiment-latency}), where the box indicates the median and
the 25th and 75th percentile scores, while the whiskers indicate variability
outside the lower and upper quartiles.
The medians are marked by the white lines between boxes.
The outliers (probably due to the garbage collection in phones) are not shown.

As indicated in Figure~\ref{fig:experiment-latency}, the \textsl{read} latency
is significantly reduced using the 2AM algorithm which completes
each \textsl{read} in one round-trip.
In the case of 5 replicas, the reduction of the latency is about 29\%.

Figure~\ref{fig:experiment-latency} also shows that the issue rate of
operations on each client has \emph{little} impact on the \textsl{read} latency.
This is due to the fact that in both algorithms, \textsl{reads} or
\textsl{writes} proceed independently, especially without waiting for
each other (the cases of ``rate = 5'', ``rate = 20'', and ``rate = 100'' are
thus not shown).
On the other hand, the more the replicas are involved, the higher
the \textsl{read} latency is incurred.
This is because each \textsl{read} needs to contact all the replicas and waits
for acknowledgments from a majority of them.
%%%%%%%%%%%%%%%%%%%%%%%%
\subsection{Experimental Result 2: Atomicity Violations}

To measure the proportion of the atomicity violations incurred in the 2AM
algorithm, we count the number of \textsl{read} operations ($\mathrm{\#R}$) and the occurrences of
concurrency patterns ($\mathrm{\#CP}$) and read-write patterns ($\mathrm{\#RWP}$).
Because each concurrency pattern or each read-write pattern is associated with
some \textsl{read} operation $r$, we are concerned with the following
quantities:
{\small
\[
  \mathrm{P(CP) = \frac{\#CP}{\#R}}, \mathrm{P(RWP|CP) =
  \frac{\#RWP}{\#CP}}, \mathrm{P(ONI) = \frac{\#RWP}{\#R}}
\]
}

In this manner, the proportion of ``old-new inversions" (and
thus the violations of atomicity) $\mathrm{P(ONI)}$ equals the product of
$\mathrm{P(CP)}$ and $\mathrm{P(RWP|CP)}$:
\begin{align}	\label{eqn:oni-product}
  \mathrm{P(ONI)} = \mathrm{P(CP)} \cdot \mathrm{P(RWP|CP)}.
\end{align}
Notice that this is not the case in theory, according to
Equation~(\ref{equ:oni-factors}).
Therefore, Equation~(\ref{eqn:oni-product}) is a practical approximation to
Equation~(\ref{equ:oni-factors}) in theory, without going into the details of
conditioning on $\mathrm{R' = m}$.
The feasibility of such an approximation will be justified by
the experimental results presented shortly, in the sense that the
key observations drawn from the numerical results based on the equations in
theory fit well with the empirical data and Equation~(\ref{eqn:oni-product}).

\begin{table*}[!t]
  \renewcommand{\arraystretch}{1.2}
  \caption{The numbers and proportions of concurrency patterns and read-write
  patterns (\emph{replication factor = 5)}.}
  \label{tbl:cp-rwp-oni-replica5}
  \centering
  \begin{tabular}{|>{\boldmath$}c<{$}||>{$}c<{$}|>{$}c<{$}|>{$}c<{$}||>{$}c<{$}|>{$}c<{$}|>{$}c<{$}|>{$}c<{$}|}
    \hline
    \text{\bfseries \# async (ms)}
%     & \text{\bfseries \# \textsl{read} operations}
%     & \text{\bfseries \# concurrency patterns}
%     & \text{\bfseries \# read-write patterns}
%     & \text{\bfseries \# \textsl{read} operations}
    & \begin{tabular}[c]{@{}c@{}}\text{\bfseries\# \textsl{read}}\\
    \text{\bfseries operations}\end{tabular}
    & \begin{tabular}[c]{@{}c@{}}\text{\bfseries\# concurrency}\\
    \text{\bfseries patterns}\end{tabular}
    & \begin{tabular}[c]{@{}c@{}}\text{\bfseries\# read-write}\\
    \text{\bfseries patterns}\end{tabular}
    & \text{\bfseries $\textrm{P(CP)}$} & \text{\bfseries P(RWP${\bf\mid}$CP)}
    & \text{\bfseries $\textrm{P(ONI)}$}
    \\ \hline \hline
    10 & 800, 000 & 269, 061 & 47 & 0.336326 & 0.000174682 & 0.00005875 \\
    \hline
    20 & 800, 000 & 306, 274 & 44 & 0.382843 & 0.000143662 & 0.000055 \\
    \hline
    50 & 800, 000 & 428, 344 & 44 & 0.53543 & 0.000102721 & 0.000055 \\
    \hline
    100 & 800, 000 & 549, 102 & 83 & 0.686378 &¡¡0.000151156 & 0.00010375 \\
    \hline
    200 & 800, 000 & 627, 814 & 100 & 0.784768 & 0.000159283 & 0.000125 \\
    \hline
  \end{tabular}
\end{table*}

\begin{table*}[!t]
  \renewcommand{\arraystretch}{1.2}
  \caption{The numbers and proportions of concurrency patterns and
  read-write patterns (\emph{async = 50 ms}).}
  \label{tbl:cp-rwp-oni-async50}
  \centering
  \begin{tabular}{|>{\boldmath$}c<{$}||>{$}c<{$}|>{$}c<{$}|>{$}c<{$}||>{$}c<{$}|>{$}c<{$}|>{$}c<{$}|>{$}c<{$}|}
    \hline
    \text{\bfseries \# replicas}
%     & \text{\bfseries \# \textsl{read} operations}
%     & \text{\bfseries \# concurrency patterns}
%     & \text{\bfseries \# read-write patterns}
%     & \text{\bfseries \# \textsl{read} operations}
    & \begin{tabular}[c]{@{}c@{}}\text{\bfseries\# \textsl{read}}\\
    \text{\bfseries operations}\end{tabular}
    & \begin{tabular}[c]{@{}c@{}}\text{\bfseries\# concurrency}\\
    \text{\bfseries patterns}\end{tabular}
    & \begin{tabular}[c]{@{}c@{}}\text{\bfseries\# read-write}\\
    \text{\bfseries patterns}\end{tabular}
    & \text{\bfseries $\textrm{P(CP)}$} & \text{\bfseries P(RWP${\bf\mid}$CP)}
    & \text{\bfseries $\textrm{P(ONI)}$}
    \\ \hline \hline
    2 & 200, 000 & 66, 985 & 0 & 0.334925 & 0 & 0 \\ \hline
    3 & 400, 000 & 192, 902 & 83 & 0.482255 & 0.00043027 & 0.0002075 \\ \hline
    4 & 600, 000 & 280, 091 & 6 & 0.466818 & 0.0000214216 & 0.00001 \\ \hline
    5 & 800, 000 & 428, 344 & 44 & 0.53543 & 0.000102721 & 0.000055 \\ \hline
  \end{tabular}
\end{table*}

Due to the limited space, Tables~\ref{tbl:cp-rwp-oni-replica5} and
\ref{tbl:cp-rwp-oni-async50} summarize part of the experimental results (also
shown in Figure \ref{fig:experiment-oni}).
In Table~\ref{tbl:cp-rwp-oni-replica5}, the replication factor is 5 (thus the
number of \textsl{read} operations is 800, 000) and the parameter of
async varies from 10 ms to 200 ms.
In Table~\ref{tbl:cp-rwp-oni-async50}, the parameter of async is 50 ms and the
replication factors vary from 2 to 5.
% Due to the limited space, we only show the cases of ``replication factor = 5''
% (Table \ref{tbl:cp-rwp-oni-replica5}) and ``async = 50 ms'' (Table \ref{tbl:cp-rwp-oni-async50}).
%  The occurrences of concurrency patterns, read-write patterns, and
% the proportion of the latter are shown in Figure~\ref{fig:experiment-oni}.
% They are all measured under various degrees of asynchrony and different replica
% factors.
As shown in Table~\ref{tbl:cp-rwp-oni-replica5}, the higher the degree of
asynchrony is, the more concurrency patterns there are.
On the other hand, the \emph{number} of occurrences of concurrency patterns
grows as the replication factor increases (Table \ref{tbl:cp-rwp-oni-async50}).
Accordingly, the \emph{proportion} of concurrency patterns $\mathrm{P(CP)}$
also increases along with the replication factor, as implied by
Equation~(\ref{equ:concurrency-pattern-summation}).

For the number of read-write patterns, the experimental results exhibit three
features.
First, no read-write patterns (and thus no ``old-new inversions'') arise in
only 2 replicas.
This is because both \textsl{read} and \textsl{write} operations are
required to contact both replicas to complete.
Second, there are fewer read-write patterns in the case of 4 replicas than
those in the case of 3 or 5 replicas.
In the case of 4 replicas, each \textsl{read} contacts 3 replicas according
to the mechanism of the majority quorum system, accounting for 75\% of them, and
gains more opportunities to obtain the latest data version.
For 3 or 5 replicas, the majorities account for 66.7\% and 60\%, respectively.
(Notice that the majority accounts for 100\% in the case of 2 replicas.)
Third, Table~\ref{tbl:cp-rwp-oni-replica5} shows that the degree of asynchrony
also contributes to the occurrences of \mbox{read-write} patterns since it may
lead to out-of-order message delivery in the timed balls-into-bins model
(Section~\ref{subsection:quantifying-read-write-pattern}).

One of the most important observations concerning these experiments is that they
have confirmed our theoretical analysis in Section~\ref{subsection:analysis}.
First, the proportion of ``old-new inversions'' $\mathrm{P(ONI)}$
is quite small (less than 0.1\text{\textperthousand} in most executions),
demonstrating that \mbox{2-atomicity} is ``good enough'' in data storage systems
regarding the violations of atomicity.
More importantly, the proportion of read-write patterns among concurrency
patterns $\mathrm{P(RWP|CP)}$ is much less than that of concurrency
patterns $\mathrm{P(CP)}$ themselves.
Namely, although concurrency patterns appear frequently (e.g.,
accounting for more than $50\%$ in the setting of 5 replicas and 50 ms async),
only a quite small portion of them satisfy the read-write semantics of
read-write pattern (Definition \ref{def:read-write-pattern}) to constitute the
``old-new inversions'' (e.g., about 0.1\text{\textperthousand} in the same
setting).
It follows that the read-write patterns dominate in guaranteeing such rare
atomicity violations incurred in the 2AM algorithm.

In conclusion, the experimental results (which have confirmed the theoretical
analysis) show that 2-atomicity and the 2AM algorithm are ``good enough''
in distributed storage systems, by achieving low latency, bounded staleness, and
rare atomicity violations.
%%%%%%%%%%%%%%%%%%%%%%%%%%%%%%%%%%%%%%%%%%%%%%%%%
\section{Related Work}	\label{section:related-work}

We divide the related work into three categories: consistency/latency tradeoff,
complexity of emulating atomic registers, and quantifying weak consistency.

\emph{Consistency/latency tradeoff.}
Designing distributed storage systems involve a range of tradeoffs among, for
instance, consistency, latency, availability, and fault-tolerance.
The well-known CAP theorem \cite{Brewer00} states that it is impossible for any
distributed data storage system to achieve consistency, availability, and
network-partition tolerance simultaneously.
% This fundamental tradeoff has driven many modern distributed storage systems
% \cite{Amazon07, Cooper08} to sacrifice strong data consistency for high
% availability, in the presence of network partition.
More recently, another tradeoff --- between consistency and
latency --- has been considered more influential on the designs of distributed
storage systems, as it is present at all times during system operation
\cite{Abadi12}.

In this paper we study the consistency/latency tradeoff and propose the notion
of almost strong consistency as a better balance option for it.

\emph{Complexity of emulating atomic registers.}
The ABD algorithm for atomicity \cite{abd-Attiya95-jacm} \cite{Attiya10}
emulates the atomic, single-writer multi-reader registers in unreliable,
asynchronous networks, given that a minority of nodes may fail.
It requires each \textsl{read} to complete in two round-trips.
Dutta \emph{et al.} \cite{Dutta04} proved that it is
\emph{impossible} to obtain a \emph{fast} emulation, where both \textsl{reads}
and \textsl{writes} complete in \emph{one} \mbox{round-trip} (i.e., low latency
in our terms).
Georgiou \emph{et al.} \cite{Georgiou08} studied the semi-fast emulations (of
atomic, single-writer multi-reader registers) where \emph{most}
\textsl{reads} complete in one \mbox{round-trip}.
Guerraoui \emph{et al.} considered the best-cases complexity,
assuming synchrony, no or few failures, and absence of \textsl{read/write}
contention.
In this situation, fast emulations do exist \cite{Guerraoui07}.

We investigate the notion of almost strong consistency in terms of
\mbox{2-atomicity}, namely, to emulate \mbox{2-atomic}, single-writer
multi-reader registers.
Our 2AM algorithm completes both \textsl{reads} and \textsl{writes} in one
round-trip.

\emph{Quantifying weak consistency.}
Weak consistency can be quantified from four perspectives: data versions,
randomness, timeliness, and numerical values.
Modern distributed storage systems often settle for weak consistency and allow
\textsl{reads} to obtain data of stale versions \cite{amazon-Vogels07-sosp},
\cite{pnuts-Cooper08-vldb}.
The semantics of $k$-atomicity \cite{Aiyer05} guarantees that the data returned
is of a bounded staleness.
Without guarantee of bounded staleness, random registers \cite{Lee05} provide
a probability distribution over the set of out-of-date values that may be
returned.
Using PBS (Probabilistically Bounded Staleness) \cite{Bailis12}, one can
obtain the probability of reading one of the latest $k$ versions of a data item.
% This quantitative analysis is conducted on the partial, non-overlapping quorum
% systems \cite{Vogels09}, \cite{Aiyer05}.
Timed consistency models \cite{Torres-Rojas99} require \textsl{writes}
to be globally visible within a period of time.
PBS \cite{Bailis12} also calculates the probability of reading a \textsl{write}
$t$ seconds after it returns.
TACT \cite{Yu02}, a continuous consistency model, integrates the metric on
numerical error with staleness.

The \mbox{2-atomicity} (and almost strong consistency) semantics integrates
bounded staleness of versions with randomness.
Our 2AM algorithm completes each \textsl{read} in one round-trip, in
contrast to that of \mbox{$k$-atomicity} \cite{Aiyer05}.
It differs from random registers \cite{Lee05} and PBS \cite{Bailis12} in two
aspects:
First, it provides guarantee of deterministically bounded staleness.
% of versions.
Second, the rate of violations is quantified with respect to
atomicity instead of regularity (as in \cite{Lee05} and \cite{Bailis12}), which
is more challenging since we shall deal with \emph{concurrent} operations.
% (It turns out that for the single-writer case, our \mbox{2-atomicity} is
% equivalent to regularity \cite{interprocess-Lamport86-dc}).
To do this, we propose a stochastic queueing model for analyzing the
concurrency pattern first and then a timed balls-into-bins model for
analyzing the \mbox{read-write} pattern.
%%%%%%%%%%%%%%%%%%%%%%%%%%%%%%%%%%%%%%%%%%%%%%%%%
\section{Conclusion and Future Work}	\label{section:conclusion}

In this paper we propose the notion of \emph{almost strong consistency}
as a better balance option for the consistency/latency tradeoff.
It provides both deterministically bounded staleness of data versions for each
\textsl{read} and probabilistic quantification on the rate of ``reading stale
values'', while achieving low latency.
% On one hand, it provides deterministically bounded staleness of data versions
% for each \textsl{read}.
% On the other hand, it provides probabilistic quantification on the rate
% of ``reading stale values''.
In the context of distributed storage systems, we investigate almost strong
consistency in terms of \mbox{\emph{2-atomicity}}.
Our 2AM (\mbox{2-Atomicity} Maintenance) algorithm completes both
\textsl{reads} and \textsl{writes} in \emph{one} communication round-trip,
and guarantees that each \textsl{read} obtains the value of within the latest 2
versions.
% We present the 2AM (\mbox{2-Atomicity} Maintenance) algorithm which allows
% both \textsl{read} and \textsl{write} operations to complete in one round-trip.
We also quantify the rate of atomicity violations incurred in the 2AM
algorithm, both analytically and experimentally.

We identify three problems for future work.
First, it is worthwhile to conduct more intensive simulations or experiments, in
order to reveal the key parameters and guiding principles for distributed
storage system design.
Second, we plan to study \mbox{2-atomic}, \emph{multi-writer} multi-reader
registers.
One key problem is whether they admit implementations which complete
both \textsl{reads} and \textsl{writes} in one \mbox{round-trip}.
Finally, we hope to extend the notion of almost strong consistency from
shared registers to snapshot objects.

%%%%%%%%%%%%%%%%
%\balance
%%%%%%%%%%%%%%%%%%%%%%%%%%%%%%%%%%%%%%%%%%%%%%%%%
\section{Acknowledgments}

This work is supported by the National 973 Program of China (2015CB352202) and
the National Science Foundation of China (61272047, 91318301, 61321491).
The authors thank the users from MathOverflow \cite{mathoverflow1} \cite{mathoverflow2}
for helpful discussions on the calculations in Section
\ref{section:quantifying}.
%%%%%%%%%%%%%%%%%%%%%%%%%%%%%%%%%%%%%%%%%%%%%%%%%
\bibliographystyle{abbrv}
\bibliography{asc}
%%%%%%%%%%%%%%%%%%%%%%%%%%%%%%%%%%%%%%%%%%%%%%%%%
\clearpage % new page

\begin{appendix}

\section{Calculations of $\boldsymbol{\mathbb{P}(E_{N-1,m})}$ in Section 4.1}
% \ref{subsection:quantifying-concurrency-pattern}}
\label{appendix:calculation}

In this section, we compute the probability of the event, denoted $E_{N-1,m}$,
that there are totally $m$ \textsl{read} operations in $N-1$ queues (besides
$Q_0$) which finish during the time period $[w_{st}, r_{st}]$ (\emph{Step 3}
in Section \ref{subsection:quantifying-concurrency-pattern}).

We first consider a single queue.
Let $D$ be a random variable denoting the number of operations in \emph{one}
particular queue which finish during the time period $r_{st} - w_{st}$.
Its probability distribution $\mathbb{P}(D = d)$ is given in Appendix
\ref{appendix:single-queue}.
Then, we take into account all the $N-1$ ($N>1$) queues, besides $Q_0$.
The calculations of $\mathbb{P}(E_{N-1,m})$ are given in Appendix
\ref{appendix:balls-into-bins}.

%%%%%%%%%%%%%%%%%%%%%%%%%%%%%
\subsection{Calculations of $\boldsymbol{\mathbb{P}(D = d)}$}
\label{appendix:single-queue}

Let $D$ be a random variable denoting the number of operations in \emph{one}
particular queue which finish during the time period $L = r_{st} - w_{st}$.
To compute its probability distribution, we condition on whether $w$ sees this
queue as empty (denoted as an event $E_{\emptyset}$) or not (denoted as an event
$E_{\neq \emptyset}$).

1) If it sees this queue empty (with probability $a_0 = \frac{\mu}{\mu +
\lambda}$), then the number of departures, during the time period $r_{st} -
w_{st}$, has the conditional distribution:
\begin{align*}
  \mathbb{P}&(D = d \mid E_{\emptyset}) \\
  &= \left\{
    \begin{array}{ll}
     \mathbb{P}(L < A_0 + S_0) &\textrm{if $d = 0$} \\ \\
     \mathbb{P} \big(\smashoperator[r]{\sum_{i=1}^{d}} (A_i + S_i) \le L <
     \sum_{i=1}^{d + 1} (A_i + S_i) \big) &\textrm{if $d \geq 1$}
    \end{array} \right. \\ \\
  &= \left\{
    \begin{array}{ll}
      \frac{2 \lambda + \mu}{2 (\lambda + \mu)} &\textrm{if $d = 0$} \\ \\
      (1 - \frac{1}{2} \frac{\mu}{\mu + \lambda}) (\frac{1}{2}
      \frac{\mu}{\mu + \lambda})^{d} \qquad &\textrm{if $d \geq 1$}
%       \frac{\mu^{d}(2 \lambda + \mu)}{2^{d + 1} (\lambda + \mu)^{d + 1}}
%       &\textrm{if $d \geq 1$}
    \end{array} \right.
\end{align*}

where $A_i$ are independent and identically distributed
(iid) exponential random variables with parameter $\lambda$ corresponding to the
inter-arrival times of operations in the other queue, and $S_i$ are iid
exponential random variables with parameter $\mu$ corresponding to the service
time of these operations.

Here we briefly demonstrate the calculation of
\[
  \mathbb{P} \big(\smashoperator[r]{\sum_{i=1}^{d}} (A_i + S_i) \le L <
  \sum_{i=1}^{d + 1} (A_i + S_i) \big) \quad (\textrm{when } d \geq 1).
\]
For convenience, we write
\[
  R_{d} \triangleq \smashoperator[r]{\sum_{i=1}^{d}} (A_i + S_i) \textrm{ and }
  R_{d+1} \triangleq \sum_{i=1}^{d + 1} (A_i + S_i).
\]

As $L$ is an exponential random variable with parameter $\lambda$ and is
independent of $R_d$, we have
\[
  \mathbb{P}(R_d \le L) = \int P(L \geq x) d P_{R_d}(x) =
  \mathbb{E}(e^{-\lambda R_d}).
\]
It follows from the independence assumptions that,
\begin{align*}
  \mathbb{P} &(R_d \le L < R_{d+1}) = \mathbb{E}(e^{-\lambda R_d} - e^{-\lambda
  R_{d+1}}) \\
  &= \prod_{i=1}^{d} \mathbb{E}(e^{-\lambda (A_i + S_i)}) - \prod_{i=1}^{d+1}
  \mathbb{E}(e^{-\lambda (A_i + S_i)}) \\
  &= \big( \mathbb{E}(e^{-\lambda(A_1 + S_1)}) \big)^{d} -
  \big( \mathbb{E}(e^{-\lambda(A_1 + S_1)}) \big)^{d+1} \\
  &= (\frac{1}{2} \frac{\mu}{\mu + \lambda})^{d} - (\frac{1}{2} \frac{\mu}{\mu +
  \lambda})^{d+1} \\
  &= (1 - \frac{1}{2} \frac{\mu}{\mu + \lambda}) (\frac{1}{2} \frac{\mu}{\mu +
  \lambda})^{d}
\end{align*}

2) Similarly, if it sees this queue full (with probability $a_1 =
\frac{\lambda}{\mu + \lambda}$), we have
\begin{align*}
  \mathbb{P} &(D = d \mid E_{\neq \emptyset}) \\
  &= \left\{
    \begin{array}{ll}
      \mathbb{P}(L < S_0) &\textrm{if $d = 0$} \\ \\
      \mathbb{P} \big(\smashoperator[r]{\sum_{i=1}^{d}} S_i + \sum_{i=1}^{d-1}
      A_i \le L
      \\ \qquad \qquad < \smashoperator[r]{\sum_{i=1}^{d+1}} S_i +
      \sum_{i=1}^{d} A_i \big) &\textrm{if $d \geq 1$} \end{array} \right. \\ \\
  &= \left\{
    \begin{array}{ll}
      \frac{\lambda}{\mu + \lambda} \qquad &\textrm{if $d = 0$}  \\ \\
      \frac{\mu + 2\lambda}{\mu + \lambda} (\frac{1}{2} \frac{\mu}{\mu +
      \lambda})^{d} \qquad &\textrm{if $d \geq 1$}
    \end{array} \right.
\end{align*}

Using the law of total probability, we obtain
\begin{align} \label{eqn:D=d}
  \mathbb{P}&(D = d)
  \nonumber\\
  &= \frac{\mu}{\mu + \lambda} \mathbb{P}(D = d \mid E_{\emptyset}) +
  \frac{\lambda}{\mu + \lambda} \mathbb{P}(D = d \mid E_{\neq \emptyset})
  \nonumber\\
  &= \left\{
    \begin{array}{ll}
      \frac{1}{2} \big( 1 + (\frac{\lambda}{\mu + \lambda})^2 \big) &\textrm{if $d = 0$}
      \\ \\
      \frac{(2\lambda + \mu)^2}{2(\mu + \lambda)^2} (\frac{1}{2}
      \frac{\mu}{\mu + \lambda})^{d} \qquad &\textrm{if $d \geq 1$} \end{array}
      \right.
\end{align}
%%%%%%%%%%%%%%%%%%%%%%%%%%%%%%%%%
\subsection{Calculations of $\boldsymbol{\mathbb{P}(E_{N-1,m})}$}
\label{appendix:balls-into-bins}

Taking into account all the $N-1$ ($N > 1$) queues, besides $Q_0$, we can
compute the probability of the event, denoted $E_{N-1, m}$, that there are
exactly $m$ \textsl{read} operations which finish during $L = r_{st} - w_{st}$
by modeling it as a balls-into-bins problem.

There are $N-1$ bins, labeled with $1, 2, \ldots, N-1$. Let $X_i$ be a random
variable denoting the number of balls contained in the $i$-th bin. The
collection of random variables $X_i$ is independent and identically distributed,
with the same probability distribution

\begin{displaymath}
  p_x = \mathbb{P}(X_i = x) = \left\{ \begin{array}{ll}
    \frac{1}{2} \left( 1 + (\frac{\lambda}{\mu + \lambda})^2 \right) &\textrm{if
    $x = 0$}
    \\ \\
    \frac{(2 \lambda + \mu)^2}{2\left(\mu + \lambda\right)^2} \cdot
    \left(\frac{1}{2} \frac{\mu}{\mu + \lambda} \right)^{x}& \textrm{if $x \geq 1$}\\
  \end{array} \right.
\end{displaymath}

We want to compute the probability of the event, denoted $E_{N-1,m}$, that there
are in total $m$ balls in these $N-1$ bins. For convenience, we write

\[
  r \triangleq \frac{(2 \lambda + \mu)^2}{2(\mu + \lambda)^2} \textrm{ and } s
  \triangleq \frac{1}{2} \frac{\mu}{\mu + \lambda}.
\]

% There are two cases according to whether there are empty bins.

% \begin{enumerate}
%   \item There are no empty bins (denoted as an event $E_{\forall X_i > 0}$). In
%   this case, we are partitioning integer $m$ into a sum of $n$ \emph{positive} integers.
%   There are ${m-1 \choose n-1}$ ways of partitions. For each partition
%
%   \[
%     m = m_1 + m_2 + \cdots + m_n \quad (m_i > 0, 1 \le i \le n),
%   \]
%
%   the probability that the $i$-th bin contains $m_i$ balls is
%
%   \[
%     (r \cdot s^{m_1}) (r \cdot s^{m_2}) \cdots (r \cdot s^{m_n}) = r^n s^m
%   \]
%
%   Thus, the probability of the case there are no empty bins is
%
%   \[
%     \mathbb{P}(E_{\forall X_i > 0}) = {m-1 \choose n-1} r^n s^m
%   \]

%   \item There exist empty bins (denoted as an event $E_{\exists X_i = 0}$).
  First assume $m > 0$.
  Let $K$ be a random variable denoting the number of empty bins. Suppose there
  are $k$ ($0 \le k \le N-2$) empty bins (i.e., $K = k$).
  In this case, we are partitioning integer $m$ into a sum of $N-1$ integers
  such that $k$ of them are 0 and $N-1-k$ of them are positive.
  There are $\binom{N-1}{k} \binom{m-1}{N-k-2}$ ways of partitions.
  For each partition
  \[
    m = m_1 + m_2 + \cdots + m_k + m_{k+1} + m_{k+2} + \cdots + m_{N-1}
  \]
  such that $m_i = 0$ for $1 \le i \le k$ and $m_i > 0$ for $k+1 \le i \le N-1$,
  the probability that the $i$-th bin contains $m_i$ balls is
  \[
    p_0^{k} \cdot (r \cdot s^{m_{k+1}}) (r \cdot s^{m_{k+2}}) \cdots (r \cdot
    s^{m_n}) = p_0^{k} \cdot r^{N-1-k} \cdot s^{m}
  \]

  Therefore, the probability that there exist $k$ ($0 \le k \le N-2$) empty bins
  is
 \begin{align*}
    \mathbb{P}(E_{N-1,m}, &K = k)
    \\
    &= \binom{N-1}{k} \binom{m-1}{N-k-2} p_0^{k} r^{N-k-1} s^{m}
 \end{align*}

 Summing over all $k$ yields (recall that $m > 0$)
 \begin{align*}
  \mathbb{P}(E_{N-1,m}) &= \sum_{k=0}^{N-2} \mathbb{P} \big( (E_{N-1,m}, K = k)
  \big)
  \\
  	&= \sum_{k=0}^{N-2} \binom{N-1}{k} \binom{m-1}{N-k-2} p_0^{k} r^{N-k-1}
  	s^{m}
 \end{align*}

 For the special case $m = 0$, we have
 \[
   \mathbb{P}(E_{N-1,0}) = p_0^{N-1}
 \]

%%%%%%%%%%%%%%%%%%%%%%%%
\newpage % new column

\section{Calculations of $\boldsymbol{\mathbb{P} \{ r = R(w') \}}$ in \\ Section
4.2}
% \ref{subsection:quantifying-read-write-pattern}}
\label{appendix:read-write-pattern}

In this section, we compute the probability of $r$ reading from $w'$ (i.e.,
$\mathbb{P} \set{r = R(w')}$).
According to Equation~(\ref{eqn:rwp-in-text}), we shall compute both $\mathbb{P}
\set{r \neq R(w)}$ (see Appendix~\ref{appendix:rwp-r-w}) and $\mathbb{P}\set{r'
\neq R(w) \mid r \neq R(w)}$ (see Appendix~\ref{appendix:rwp-rprime-w}).
For the latter probability, we also introduce a slightly generalized timed
balls-into-bins model in Appendix~\ref{appendix:rwp-generalized-model}.

%%%%%%%%%%%%
\begin{comment}
\subsection{Upper bound for $\mathbb{P}\{ r = R(w') \}$}
\label{appendix:rwp-upper-bound}

The following gives the calculation of the probability that the \textsl{read}
$r$ reads from the \textsl{write} $w'$ (i.e., $r = R(w')$):
\begin{align*}
  \mathbb{P} &\{ r = R(w') \} \\
  &= \mathbb{P} \{ r \neq R(w'')_{[w'' \prec w']}
  	\land r \neq R(w) \land r \neq R('w)_{[w \prec ('w)]} \} \\
  &\pushright{ (r \textrm{ does not read from other
  \textsl{writes}; see Figure~\ref{fig:old-new-inversion}}) }
  \\
  &\pushright{([w'' \prec w'] \text{ means: for all the \textsl{writes} } w''
  \text{ that precede } w')} \\
  &\pushright{([w \prec ('w)] \text{ means: for all the \textsl{writes} } 'w
  \text{ that are preceded by } w)} \\
  &= \mathbb{P} \{ r \neq R(w) \land r \neq R('w)_{[w \prec ('w)]} \} \\
  &\pushright{ (\textrm{due to } w'' \prec w' \prec r) } \\
  &= \mathbb{P} \{ r \neq R(w) \} \cdot \mathbb{P} \{ r \neq R('w)_{[w \prec
  ('w)]} \} \\
  &\pushright{ (\textrm{independence assumption}) } \\
  &\le \mathbb{P} \{r \neq R(w) \}
\end{align*}
\end{comment}
%%%%%%%%%%%%
\subsection{Calculations of $\boldsymbol{\mathbb{P} \{ r \neq R(w) \}}$ in the
timed balls-into-bins model}
\label{appendix:rwp-r-w}

Let $q = \lfloor n/2 \rfloor + 1$.
Denote the delay times for each ball from robot $R_1$ (corresponding to $w$)
sent to each bin $B_i$ by $D'_i$ and the delay times for each ball from robot
$R_2$ (corresponding to $r$) sent to each bin $B_i$ by $D_i$.
Let $M_m = \max \{D_1, D_2, \ldots, D_m \}$.
By symmetry,
\[
  \mathbb{P}(E) = \binom{n}{q}\,\mathbb{P}\left(E, B = \{ 1,\ldots,q
  \}\right).
\]
If $D_1 = M_q$, we shall compute
\begin{align*}
  I_1 &\equiv \mathbb{P}\{D_1^\prime > t + M_q, D_2^\prime > t + D_2,
  \ldots, D_q^\prime > t + D_q, \\
  &D_{q+1} > M_q, D_{q+2} > M_q, \ldots, D_n > M_q\}.
\end{align*}
Conditioning on $M_q = D_1, D_2, \ldots, \textrm{ and } D_q$ and using the
independence assumptions, we obtain:
\begin{align*}
  I_1 = \int_{0}^{\infty} &\idotsint_V e^{-\lambda_{w}(t+s)} \left(\prod_{i=2}^q
  e^{-\lambda_{w}(t+x_i)}\right) e^{-\lambda_{r}(n-q)s} \\
    &f(s,x_2,\ldots,x_q) \,dx_2\ldots dx_q \, ds,
\end{align*}
where
\[
  V = [0,s]^{q-1} \subseteq \mathbb{R}^{q-1},
\]
and
\[
  f(s, x_2, \ldots, x_q) = \mathbf{1}_{[0,\infty)}(s)
  \lambda_{r}e^{-\lambda_{r}s} \prod_{i=2}^{q} \lambda_{r}e^{-\lambda_{r}x_i}
  \mathbf{1}_{[0,s]}(x_i).
\]
Here, $\mathbf{1}_{[0,\infty)}(s)$ and $\mathbf{1}_{[0,s]}(x_i)$ are
indicator functions.

The integral over $x_i$ is:
\begin{align*}
  \int_{0}^{s} e^{-\lambda_{w} (t + x_i)} \, e^{-\lambda_{r} x_i} \, dx_i =
  e^{-\lambda_{w} t} \frac{1-e^{-(\lambda_{w}+\lambda_{r})s}} {\lambda_{w} +\lambda_{r}}.
\end{align*}
By independence of all $x_i$'s, we carry out all the $x_i$ integrals and obtain
\begin{align*}
  I_1 &= e^{-q\lambda_{w}t} \lambda_{r}^q \\
  &\cdot \int_0^\infty e^{-(\lambda_{w}+\lambda_{r})s}\left(\frac{1-e^{-(\lambda_{w}+\lambda_{r})s}}
  {\lambda_{w} +\lambda_{r}} \right)^{q-1} e^{-\lambda_{r}(n-q)s}\,ds
\end{align*}
Making the substitution $y=1-e^{-(\lambda_{w}+\lambda_{r})s}$ yields
\[
  I_1=e^{-q \lambda_{w} t}\alpha^q\,B(q,\alpha (n-q)+1),
\]
where $\alpha = \frac{\lambda_{r}}{\lambda_{w} + \lambda_{r}}$ and $B$ denotes
the Beta function.

Finally, by symmetry, the cases $D_2 = M_q, \ldots, \textrm{and } D_q = M_q$
give the same result, so that
\begin{align*}
  \mathbb{P}(E)&= q \binom{n}{q} e^{-q\lambda_{w}t}\alpha^q \,
  B(q,\alpha(n-q)+1)
  \\
  &= e^{-q\lambda_{w}t} \frac{\alpha^q\,B(q,\alpha(n-q)+1)}{B(q,n-q+1)}.
\end{align*}
%%%%%%%%%%%%
\subsection{Generalized timed balls-into-bins model for the case of
$\boldsymbol{r' \neq R(w)}$ conditioning on $\boldsymbol{r \neq R(w)}$}
\label{appendix:rwp-generalized-model}

Given $r \neq R(w)$ and $r' \prec r$, some messages from $w$ are \emph{known} to
reach the replicas later than the time $r'$ has collected enough acknowledgments
and finished.
To calculate $\mathbb{P} \{ r' \neq R(w) \mid r \neq R(w) \}$, we
introduce a slightly \emph{generalized timed balls-into-bins model}.
In the generalized model, at time $t$, robot $R_2$ picks $p$ ($0 < p
\leq n$) bins uniformly at random (without replacement) and sends a ball to
each of them, instead of sending a ball to each of the $n$ bins as before.
The remaining $(n-p)$ unsent balls are used to model the messages that arrive
late.

For the case of $\set{r' \neq R(w) \mid r \neq R(w)}$, we consider the
generalized model in which robots $R_1$ and $R_2$ represent $r'$ and $w$, respectively.
We assume that the random variable $D_{r}$ (resp. $D_{w}$) for time delay is
exponentially distributed with rate $\lambda_r$ (resp. $\lambda_w$).
It remains to calculate the expected time lag between the events that $r'$
and $w$ are issued, i.e., $\mathbb{E} \{ w_{st} - r'_{st} \}$.
This is challenging because there may be more than one such $r'$ following the
concurrency pattern (Definition~\ref{def:concurrency-pattern}) in a single
process.
Nevertheless the probability that there are $k$ ($k \geq 1$) such $r'$s in a
single process decreases exponentially with the ratio $\frac{\mu}{2 (\mu + \lambda)}$,
according to Equation~(\ref{eqn:D=d}) in Appendix~\ref{appendix:single-queue}.
Therefore, we focus on the simple case that there is at most one $r'$ in a
single process.
In this situation, the calculation presented shortly yields that $\mathbb{E} \{
w_{st} - r'_{st} \}= \frac{2 \lambda - \mu}{2 \lambda \mu}$.
Finally, we are interested in the time point $t'$ when exactly $q \triangleq
\lfloor n/2 \rfloor + 1$ of the $n$ bins have received the balls from $R_1$
(i.e., $r'$), and denote the set of these $q$ bins by $B$.
In terms of the \emph{generalized} timed balls-into-bins model, the case of
$\set{r' \neq R(w) \mid r \neq R(w)}$ corresponds to the event $E'$ that none of
the $q$ bins in $B$ receives a ball from $R_2$ (i.e., $w$) before it receives a ball
from $R_1$ (i.e., $r'$).

We calculate the expected time lag between the events that $r'$ and $w$ are
issued (i.e., $\mathbb{E} \{ w_{st} - r'_{st} \}$) as follows.
To this end, we first calculate the expected duration of the interval $[r'_{ft},
r_{st}]$.
Since $r'$ is required to finish between the interval $L = [w_{st}, r_{st}]$
whose length follows an exponential distribution with rate $\lambda$, and the
inter-arrival time (between $r'$ and $r$ here), denoted $I$, also follows an
exponential distribution with rate $\lambda$, we have
\begin{align*}
  \mathbb{E} \{ r_{st} - r'_{ft} \} = \mathbb{E} \{ I \mid I < L \} =
  \frac{1}{2 \lambda}.
\end{align*}

Thus, the expected time lag between the events that $r'$ and $w$ are issued is
\begin{align}	\label{eqn:t2}
  \mathbb{E} \{ w_{st} - r'_{st} \} &= \mathbb{E} \{ r_{st} - r'_{ft} \} + \mathbb{E} \{ r'_{ft} - r'_{st} \} -
  \mathbb{E} \{ r_{st} - w_{st} \}
  \nonumber\\
  &= \frac{1}{2 \lambda} + \frac{1}{u} - \frac{1}{\lambda}
  = \frac{2 \lambda - \mu}{2 \lambda \mu}.
\end{align}
\subsection{Calculations of $\boldsymbol{\mathbb{P} \{ r' \neq R(w) \mid r \neq
R(w) \}}$}
\label{appendix:rwp-rprime-w}

Let $q = \lfloor n/2 \rfloor + 1$.
Denote the delay times for each ball from robot $R_1$ (i.e., $r'$)
sent to each bin $B_i$ by $D'_i$ and the delay times for each ball from robot
$R_2$ (i.e., $w$) sent to each bin $B_i$ by $D_i$.
Let $M_q = \max \{D'_1, D'_2, \ldots, D'_q \}$.
By symmetry,
\[
  \mathbb{P}(E') = \binom{n}{q}\,\mathbb{P}\left(E', B = \{ 1,\ldots,q
  \}\right).
\]

Given $r \neq R(w)$ and $r' \prec r$, we know that $q$ balls from $w$
are bound to reach the replicas later than the time $t'$ of interest.
The other $(n - q)$ (corresponding to the parameter $p$ in the generalized
model) balls are randomly and uniformly sent into $(n - q)$ replicas, one ball
per bin.
We denote this set of $(n-q)$ replicas by $B'$.

The case $n=2$ is trivial: Since $q = n = 2$, these two balls from $w$ are bound
to reach the replicas later than the time $r'$ has collected enough
acknowledgments and returned.
Therefore, $\mathbb{P} \set{r' \neq R(w) \mid r \neq R(w)} = 1.$

Now we consider $n > 2$.
Assume $M_q = D'_1$ (without loss of generality, the corresponding bin for
$D'_1$ is denoted by $b_1$; hence $b_1 \in B$) and $k = |B \cap B'|$ ($0 \leq k \leq n-q$), we distinguish the case $b_1 \in B'$ from $b_1 \notin B'$.
Thus, we shall compute
\begin{align*}
  J_1 \equiv \sum_{k=0}^{n-q} &\Big(\mathbb{P}\{D_1 > M_q - t', D_2 > D'_2 - t',
  \ldots, D_k > D'_k - t' ,
  \\
  &\qquad D'_{q+1} > M_q, D'_{q+2} > M_q, \ldots, D'_n > M_q\}
  \\
  &+  \vphantom{\sum_{k=0}^{n-q}} \mathbb{P} \{D_2 > D'_2 - t', \ldots, D_{k+1} >
  D'_{k+1} - t' ,
  \\
  &\qquad D'_{q+1} > M_q, D'_{q+2} > M_q, \ldots, D'_n > M_q \} \Big),
\end{align*}
where $t' = \mathbb{E} \{ w_{st} - r'_{st} \} = \frac{2 \lambda - \mu}{2 \lambda
\mu}$ (see Equation~(\ref{eqn:t2}) in
Appendix~\ref{appendix:rwp-generalized-model}).

Conditioning on $M_q = D'_1, D'_2, \ldots, \textrm{ and } D'_q$ and using
the independence assumptions, we obtain:
\begin{align} \label{eqn:j1-original}
  J_1 = \sum_{k=0}^{n-q} &\left( \frac{\binom{1}{1} \binom{q-1}{k-1}
  \binom{n-q}{n-q-k}}{\binom{n}{n-q}} \right.
  \nonumber\\
  &\quad \cdot \int_{0}^{\infty} \idotsint_V
  \left(e^{\lambda_{w}(t'-s)}\right)^{\mathbf{1}_{s>t'}(s)}
  \nonumber\\
  &\qquad \cdot \left(\prod_{i=2}^{k} \left(e^{\lambda_{w}(t' -
  x'_i)}\right)^{\mathbf{1}_{x'_i > t'}(x'_i)} \right) e^{-\lambda_{r}(n-q)s}
  \nonumber\\
  &\qquad \cdot f(s,x'_2,\ldots,x'_q) \,dx'_2\ldots dx'_q \, ds
  \nonumber\\
  &+ \frac{\binom{1}{0} \binom{q-1}{k} \binom{n-q}{n-q-k}}{\binom{n}{n-q}}
  \nonumber\\
  &\quad \cdot \int_{0}^{\infty} \idotsint_V \left(\prod_{i=2}^{k+1}
  \left(e^{\lambda_{w}(t' - x'_i)}\right)^{\mathbf{1}_{x'_i > t'}(x'_i)} \right)
  \nonumber\\
  &\qquad \cdot e^{-\lambda_{r}(n-q)s}
  \nonumber\\
  &\left. \qquad \cdot f(s,x'_2,\ldots,x'_q) \,dx'_2\ldots dx'_q \, ds
  \vphantom{\frac{\binom{1}{1} \binom{q-1}{k-1}
  \binom{n-q}{n-q-k}}{\binom{n}{n-q}}}\right),
\end{align}
where
\[
  V = [0,s]^{q-1} \subseteq \mathbb{R}^{q-1},
\]
and
\[
  f(s, x'_2, \ldots, x'_q) = \mathbf{1}_{[0,\infty)}(s) \lambda_{r}e^{-\lambda_{r}s}
  \prod_{i=2}^{q} \lambda_{r}e^{-\lambda_{r} x'_i} \mathbf{1}_{[0,s]}(x'_i).
\]
Notice that $\left(e^{\lambda_{w}(t'-s)}\right)^{\mathbf{1}_{s > t'}(s)}$
denotes a piecewise function with respect to $s$:
\begin{displaymath}
 \left(e^{\lambda_{w}(t'-s)}\right)^{\mathbf{1}_{s > t'}(s)} =
  \begin{cases}
      \hfill e^{\lambda_{w}(t' - s)}    \hfill & \text{ if $s > t'$}; \\
      \hfill 1 \hfill & \text{ if $s \leq t'$}. \\
  \end{cases}
\end{displaymath}
and, similarly,
% $\left(e^{\lambda_{w}(t' - x'_i)}\right)^{\mathbf{1}_{x'_i >
% t'}(x'_i)}$ denotes a piecewise function with respect to $x'_i$:
\begin{displaymath}
 \left(e^{\lambda_{w}(t' - x'_i)}\right)^{\mathbf{1}_{x'_i > t'}(x'_i)} =
  \begin{cases}
      \hfill e^{\lambda_{w}(t' - x'_i)}    \hfill & \text{ if $x'_i > t'$}; \\
      \hfill 1 \hfill & \text{ if $x'_i \leq t'$}. \\
  \end{cases}
\end{displaymath}
For convenience, we denote the first multiple integral in $J_1$ by $J_{11}$ and
the second $J_{12}$, and focus on the calculations of $J_{11}$ in the following.
First of all, we evaluate the leftmost integral of $J_{11}$ over $s$ by breaking
it into two parts:
\begin{align} \label{eqn:j11-twoparts}
  J_{11} = \int_{0}^{t'} g(s) \, ds + \int_{t'}^{\infty} g(s) \, ds,
\end{align}
where $g(s)$ is the integrand in $J_{11}$ with respect to variable $s$.

In the first integral over $s \in [0,t']$, we have $x'_i \leq s \leq t'$ for $i =
2,3, \ldots, q$.
Thus it reduces to
\begin{align} \label{eqn:j11-0-t'}
  \int_{0}^{t'} g(s) \, ds &= \int_{0}^{t'} \idotsint_V e^{-\lambda_{r}(n-q)s}
  \mathbf{1}_{[0,t']}(s) \lambda_{r}e^{-\lambda_{r}s}
  \nonumber \\
  & \qquad \cdot \prod_{i=2}^{q} \lambda_{r}e^{-\lambda_{r}
  x'_i} \mathbf{1}_{[0,s]}(x'_i) \,dx'_2\ldots dx'_q \, ds
  \nonumber \\
  &= \lambda_{r}^{q} \int_{0}^{t'} e^{-\lambda_{r} (n-q+1) s}
  \nonumber \\
  & \qquad \cdot \left(\idotsint_V \prod_{i=2}^{q} e^{-\lambda_{r} x'_i} \,dx'_2\ldots
  dx'_q\right) \, ds
  \nonumber \\
  &= \lambda_{r} \int_{0}^{t'} e^{-\lambda_{r} (n-q+1) s}
  \left(1-e^{-\lambda_{r} s}\right)^{q-1}\, ds.
\end{align}
The second integral over $s \in [t',\infty)$ reduces to
\begin{align*}
  \int_{t'}^{\infty} g(s) \, ds &= \int_{t'}^{\infty} \idotsint_V
  e^{\lambda_{w}(t'-s)}
  \\
  &\qquad \cdot \left(\prod_{i=2}^{k} \left(e^{\lambda_{w}(t' -
  x'_i)}\right)^{\mathbf{1}_{x'_i > t'}(x'_i)} \right) e^{-\lambda_{r}(n-q)s}
%   \\
%   &\cdot \lambda_{r}e^{-\lambda_{r}s} \prod_{i=2}^{q} \lambda_{r}e^{-\lambda_{r}
%   x'_i}\, dx'_2\ldots dx'_q \, ds.
  \\
  &\qquad \cdot \mathbf{1}_{[t',\infty)}(s) \lambda_{r}e^{-\lambda_{r}s}
  \prod_{i=2}^{q} \lambda_{r}e^{-\lambda_{r} x'_i} \mathbf{1}_{[0,s]}(x'_i)
  \\
  &\qquad \cdot dx'_2\ldots dx'_q \, ds.
\end{align*}
Each of the $k-1$ integrals over $x'_i$ ($i = 2,3,\ldots, k$) is
\begin{align*}
  \int_{0}^{s} &e^{\lambda_{w} (t'-x'_i)} e^{-\lambda_{r} x'_i} \, dx'_i
  \\
  &= \int_{0}^{t'} e^{-\lambda_{r} x'_i} \, dx'_i + \int_{t'}^{s} e^{\lambda_{w}
  (t'-x'_i)} e^{-\lambda_{r} x'_i} \, dx'_i
  \\
  &= \frac{1-e^{-\lambda_{r} t'}}{\lambda_{r}} + e^{\lambda_{w} t'} \cdot
  \frac{e^{-(\lambda_{w} + \lambda_{r}) t'} - e^{-(\lambda_{w} + \lambda_{r})
  s}}{\lambda_{w} + \lambda_{r}},
\end{align*}
while each of the remaining $(q-k)$ integrals over $x'_i$ ($i = k+1, k+2,
\ldots, q$) is
\begin{align*}
  \int_{0}^{s} &e^{-\lambda_{r} x'_i} \, dx'_i = \frac{1 - e^{-\lambda_{r}
  s}}{\lambda_{r}}.
\end{align*}
Carrying out all the $x'_i$ integrals, we obtain
\begin{align} \label{eqn:j11-t'-infty}
  \int_{t'}^{\infty} &g(s) \, ds = \lambda_{r}^{q} \, e^{\lambda_{w} t'} \int_{t'}^{\infty}
  e^{-(\lambda_{w} + \lambda_{r}) s}
%   e^{\lambda_{w}(t'-s)} \, e^{- \lambda_{r} s}
  \nonumber\\
  &\cdot \left( \frac{1-e^{-\lambda_{r} t'}}{\lambda_{r}} + e^{\lambda_{w} t'} \cdot
  \frac{e^{-(\lambda_{w} + \lambda_{r}) t'} - e^{-(\lambda_{w} + \lambda_{r}) s}}{\lambda_{w} + \lambda_{r}}
  \right)^{k-1}
  \nonumber\\
  &\cdot \left(\frac{1 - e^{-\lambda_{r} s}}{\lambda_{r}}\right)^{q-k} \,
  e^{-\lambda_{r}(n-q)s} \,ds.
\end{align}
% \frac{2 - e^{- \lambda_{r} t'} - e^{\lambda_{r} t' - 2 \lambda_{r} s}}{2 \lambda_d}
Substituting Equations~(\ref{eqn:j11-0-t'}) and (\ref{eqn:j11-t'-infty})
into Equation~(\ref{eqn:j11-twoparts}) yields
\begin{align}	\label{eqn:j11}
  J_{11} &= \lambda_{r} \int_{0}^{t'} e^{-\lambda_{r} (n-q+1) s}
  \left(1-e^{-\lambda_{r} s}\right)^{q-1}\, ds
  \nonumber\\
  &\quad + \lambda_{r}^{q} \, e^{\lambda_{w} t'} \int_{t'}^{\infty}
  e^{-(\lambda_{w} + \lambda_{r}) s}
%   e^{\lambda_{w}(t'-s)} \, e^{- \lambda_{r} s}
  \nonumber\\
  &\quad\;\; \cdot \left( \frac{1-e^{-\lambda_{r} t'}}{\lambda_{r}} +
  e^{\lambda_{w} t'} \cdot
  \frac{e^{-(\lambda_{w} + \lambda_{r}) t'} - e^{-(\lambda_{w} + \lambda_{r}) s}}{\lambda_{w} + \lambda_{r}}
  \right)^{k-1}
  \nonumber\\
  &\quad\;\; \cdot \left(\frac{1 - e^{-\lambda_{r} s}}{\lambda_{r}}\right)^{q-k}
  \, e^{-\lambda_{r}(n-q)s}\,ds.
\end{align}

In the same way, we have
\begin{align}  \label{eqn:j12}
  J_{12} &= \lambda_{r} \int_{0}^{t'} e^{-\lambda_{r} (n-q+1) s}
  \left(1-e^{-\lambda_{r} s}\right)^{q-1}\, ds
  \nonumber\\
  &\quad + \lambda_{r}^{q} \, \int_{t'}^{\infty} e^{-\lambda_{r} s}
  \nonumber\\
  &\quad\;\; \cdot \left( \frac{1-e^{-\lambda_{r} t'}}{\lambda_{r}} + e^{\lambda_{w}
  t'}
  \cdot \frac{e^{-(\lambda_{w} + \lambda_{r}) t'} - e^{-(\lambda_{w} +
  \lambda_{r}) s}}{\lambda_{w} + \lambda_{r}} \right)^{k}
  \nonumber\\
  &\quad\;\; \cdot \left(\frac{1 - e^{-\lambda_{r} s}}{\lambda_{r}}\right)^{q-1-k}
  \, e^{-\lambda_{r}(n-q)s} \,ds.
\end{align}
Substituting Equations~(\ref{eqn:j11}) and~(\ref{eqn:j12}) into
Equation~(\ref{eqn:j1-original}) yields
\begin{align} \label{eqn:j1}
  J_1 &= \sum_{k=0}^{n-q} \left( \frac{\binom{q-1}{k-1}
  \binom{n-q}{n-q-k}}{\binom{n}{n-q}} J_{11} + \frac{\binom{q-1}{k}
  \binom{n-q}{n-q-k}}{\binom{n}{n-q}} J_{12} \right) \nonumber\\
  &=\lambda_{r} \int_{0}^{t'} e^{-\lambda_{r} (n-q+1) s}
  \left(1-e^{-\lambda_{r} s}\right)^{q-1}\, ds
  \nonumber\\
  &\quad + \sum_{k=0}^{n-q} \frac{\binom{q-1}{k-1}
  \binom{n-q}{n-q-k}}{\binom{n}{n-q}} \lambda_{r}^{q} \, e^{\lambda_{w} t'} \int_{t'}^{\infty}
  e^{-(\lambda_{w} + \lambda_{r}) s}
  \nonumber\\
  &\quad\;\; \cdot \left( \frac{1-e^{-\lambda_{r} t'}}{\lambda_{r}} + e^{\lambda_{w}
  t'} \cdot
  \frac{e^{-(\lambda_{w} + \lambda_{r}) t'} - e^{-(\lambda_{w} + \lambda_{r}) s}}{\lambda_{w} + \lambda_{r}}
  \right)^{k-1}
  \nonumber\\
  &\quad\;\; \cdot \left(\frac{1 - e^{-\lambda_{r} s}}{\lambda_{r}}\right)^{q-k} \,
  e^{-\lambda_{r}(n-q)s}\,ds
  \nonumber\\
  &\quad + \sum_{k=0}^{n-q} \frac{\binom{q-1}{k}
  \binom{n-q}{n-q-k}}{\binom{n}{n-q}} \lambda_{r}^{q} \, \int_{t'}^{\infty} e^{-\lambda_{r} s} \nonumber\\
  &\quad\;\; \cdot \left( \frac{1-e^{-\lambda_{r} t'}}{\lambda_{r}} + e^{\lambda_{w}
  t'}
  \cdot \frac{e^{-(\lambda_{w} + \lambda_{r}) t'} - e^{-(\lambda_{w} +
  \lambda_{r}) s}}{\lambda_{w} + \lambda_{r}} \right)^{k}
  \nonumber\\
  &\quad\;\; \cdot \left(\frac{1 - e^{-\lambda_{r} s}}{\lambda_{r}}\right)^{q-1-k}
  \, e^{-\lambda_{r}(n-q)s} \,ds.
\end{align}
Finally, by symmetry, the cases $D'_2 = M_q, \ldots, \textrm{and } D'_q = M_q$
give the same result, so that
\begin{align}	\label{eqn:p-e'-final}
  \mathbb{P}(E') = \left\{
  \begin{array}{ll}
    q \binom{n}{q} J_1 = \frac{J_1}{B(q,n-q+1)} & \textrm{if } n > 2, \\
    1 & \textrm{if } n = 2.
  \end{array}\right.
\end{align}
\end{appendix}
%%%%%%%%%%%%%%%%%%%%%%%%%%%%%%%%%%%%%%%%%%%%%%%%%
\end{document}